\documentclass[11pt]{article}

\usepackage{amsmath, amssymb, amsfonts, amsthm}

\numberwithin{equation}{section}

\newtheorem{theorem}{Theorem}[section]
\newtheorem{proposition}[theorem]{Proposition}
\newtheorem{corollary}[theorem]{Corollary}
\newtheorem{lemma}[theorem]{Lemma}

\newcommand{\rd}{{\rm d}}
\newcommand{\be}{\begin{equation}}
\newcommand{\ee}{\end{equation}}
\newcommand{\bey}{\begin{eqnarray}}
\newcommand{\eey}{\end{eqnarray}}

\newcommand{\eps}{\varepsilon}

\newcommand{\bx}{{\bf x}}
\newcommand{\by}{{\bf y}}

\newcommand{\ph}{\varphi}

\newcommand{\la}{\langle}
\newcommand{\ra}{\rangle}

\renewcommand{\a}{\alpha}

\newcommand{\cU}{{\cal U}}
\newcommand{\cM}{{\cal M}}

\newcommand{\bR}{{\mathbb R}}
\newcommand{\bC}{{\mathbb C}}
\newcommand{\bN}{{\mathbb N}}

\newcommand{\bi}{\bigskip}

\newcommand{\tr}{\mbox{Tr}}

\newcommand{\wt}{\widetilde}

\newcommand{\const}{\mathrm{const}}

\newcommand{\cF}{{\cal F}}

\newcommand{\cE}{{\cal E}}

\newcommand{\cV}{{\cal V}}
\newcommand{\cW}{{\cal W}}
\newcommand{\cK}{{\cal K}}
\newcommand{\cH}{{\cal H}}
\newcommand{\cL}{{\cal L}}
\newcommand{\cN}{{\cal N}}

\input epsf

\newcommand{\donothing}[1]{}

\usepackage{bbm}

\newcommand{\re}{\text{Re}}
\newcommand{\im}{\text{Im}}
\newcommand{\trr}{\mathrm{Tr}}

\newcommand{\UU}{\mathcal{U}_N(t;0)}
\newcommand{\UUc}{\mathcal{U}^*_N(t;0)}

\newcommand{\tUU}{\widetilde{\mathcal{W}}_N(t;0)}

\newcommand{\LL}{\mathcal{L}_N(t)}
\newcommand{\tLL}{\widetilde{\mathcal{M}}_N(t)}

\begin{document}



\title{Dynamical Collapse of Boson Stars}

\author{Alessandro Michelangeli${}^1$, Benjamin Schlein${}^2$  \\ \\
\small{1. Institute of Mathematics, LMU Munich} \\
\small{Theresienstr. 39, 80333 Munich, Germany} \\ \small{michel@math.lmu.de} \\ \\
\small{2. Institute for Applied Mathematics, University of Bonn} \\
\small{Endenicher Allee 60, 53115 Bonn, Germany} \\
\small{benjamin.schlein@hcm.uni-bonn.de}}

\date{}

\maketitle

\begin{abstract}
We study the time evolution in system of $N$ bosons with a relativistic dispersion law interacting through a Newtonian gravitational potential with coupling constant $G$. We consider the mean field scaling where $N$ tends to infinity, $G$ tends to zero and $\lambda = G N$ remains fixed. We investigate the relation between the many body quantum dynamics governed by the Schr\"odinger equation and the effective evolution described by a (semi-relativistic) Hartree equation. In particular, we are interested in the 
super-critical regime of large $\lambda$ (the sub-critical case has been studied in \cite{ES,KP}), where the nonlinear Hartree equation is known to have solutions which blow up in finite time. To inspect this regime, we need to regularize the interaction in the many body Hamiltonian with an $N$ dependent cutoff that vanishes in the limit $N\to \infty$. We show, first, that if the solution of the nonlinear equation does not blow up in the time interval $[-T,T]$, then the many body Schr\"odinger dynamics (on the level of the reduced density matrices) can be approximated by the nonlinear Hartree dynamics, just as in the sub-critical regime. Moreover, we prove that if the solution of the nonlinear Hartree equation blows up at time $T$ (in the sense that the $H^{1/2}$ norm of the solution diverges as time approaches $T$), then also the solution of the linear Schr\"odinger equation collapses (in the sense that the kinetic energy per particle diverges) if $t \to T$ and, simultaneously, $N \to \infty$ sufficiently fast. This gives the first dynamical description of the phenomenon of gravitational collapse as observed directly on the many body level. 
\end{abstract}

\section{Introduction and main results}
\label{sec:intro}

We consider systems of gravitating bosons known as {\it boson stars}. Assuming the particles to have a relativistic dispersion, but the interaction to be treated classically (Newtonian gravity), we arrive at the $N$-particle Hamiltonian 
\[ H_{\text{grav}} = \sum_{j=1}^N \sqrt{1-\Delta_{x_j}} - G \sum_{i<j}^N \frac{1}{|x_i -x_j|} \] acting on the Hilbert space $L^2_s (\bR^{3N})$, the subspace of $L^2 (\bR^{3N})$ containing all functions symmetric with respect to arbitrary permutations (here we use units with $\hbar=1$, $c=1$, and $m=1$, where $m$ denotes the mass of the bosons). 

\medskip

We are interested in the mean field limit where $N \to \infty$, $G\to 0$ so that $NG=:\lambda$ remains fixed. In other words, we are going to study a family of systems, parametrized by the number of bosons $N$, described by the $N$ particle Hamiltonian
\begin{equation}\label{eq:HN1}
H_N = \sum_{j=1}^N \sqrt{1-\Delta_{x_j}} - \frac{\lambda}{N} \sum_{i<j}^N \frac{1}{|x_i -x_j|} \, . 
\end{equation}

\medskip

The system is critical, and it behaves very differently depending on the value of the coupling constant $\lambda >0$. The criticality of the system is a consequence of the fact that the kinetic energy scales, for large momenta, like the potential energy (both scales as an inverse length). The potential energy can be made arbitrarily large (and negative) by moving the particles closer and closer together ($N$ particle in a box of volume $\ell^3$ have a potential energy of the order $N \ell^{-1}$, taking also into account the $1/N$ factor in front of the interaction energy). However, in order to localize particles in a small volume we have to pay a price in terms of kinetic energy (to localize $N$ particles within a box of volume $\ell^3$, we need an energy proportional to $N\ell^{-1}$). This simple observation implies that, for small values of the coupling constant $\lambda$, the kinetic energy dominates the potential energy, and that, for sufficiently large $\lambda$, the kinetic energy needed to bring particles together is not sufficient to compensate for the gain in the potential energy.

\medskip

For every $N \in \bN$, there exists therefore a critical coupling constant $\lambda_{\text{crit}} (N)$ such that $H_N$ is bounded below for all $\lambda < \lambda_{\text{crit}} (N)$ and such that
\[ \inf_{\psi \in L^2 (\bR^{3N})} \frac{\langle \psi, H_N \psi \rangle}{\| \psi \|^2} = - \infty \] for all $\lambda > \lambda_{\text{crit}} (N)$. It was proven in \cite{LY} that the critical constant is given, as $N \to \infty$, by the critical coupling constant for the Hartree energy functional
\begin{equation}\label{eq:enhar} \cE_{\text{Hartree}} (\ph) = \int \rd x \, \left| (1-\Delta)^{1/4} \ph (x)\right|^2 - \frac{\lambda}{2} \int \rd x \rd y \frac{|\ph (x)|^2 |\ph (y)|^2}{|x-y|}\, . \end{equation}
More precisely, it was proven in \cite{LY} that, as $N \to \infty$, $\lambda_{\text{crit}} (N) \to \lambda^{\text{H}}_{\text{crit}}$, where 
\[ \frac{1}{\lambda_{\text{crit}}^\text{H}} = \sup_{\ph \in L^2 (\bR^3), \| \ph \|=1} \frac{1}{2} \frac{
\int \rd x \rd y \, |\ph (x)|^2 |\ph (y)|^2 \, |x-y|^{-1}}{\int \rd x \, | |\nabla|^{1/2} \ph (x)|^2} \, .  \]  Note that, with this definition, $\cE_{\text{Hartree}} (\ph) \geq 0$ for all $\ph \in H^{1/2} (\bR^3)$ if $\lambda \leq \lambda^\text{H}_{\text{crit}}$ while \[ \inf_{\ph \in H^{1/2} (\bR^3) , \| \ph \| =1} \cE_{\text{Hartree}} (\ph) = -\infty \] if $\lambda > \lambda^\text{H}_{\text{crit}}$. It is also possible (see \cite{LY}) to give bounds on the fluctuations of $\lambda_{\text{crit}} (N)$ around $\lambda_{\text{crit}}^\text{H}$:
\[ \lambda_{\text{crit}}^\text{H} (1- c_1 N^{-1/3}) \leq \lambda_{\text{crit}} (N) \leq \lambda_{\text{crit}}^\text{H} (1+ c_2 N^{-1}) \,  \] for appropriate constants $c_1, c_2>0$.
The value of $\lambda_{\text{crit}}^\text{H}$ is not explicitly known. By Kato's inequality, $|x_i - x_j|^{-1} \leq  (\pi/2) |\nabla_{x_j}|$, it is easy to see that $\lambda^\text{H}_{\text{crit}} \geq (4/ \pi) \simeq 1.3$. In \cite{LT,LY}, it is also shown that $\lambda_{\text{crit}}^\text{H} \leq 2.7$.

\medskip

Let us discuss first the {\it subcritical} case $\lambda < \lambda^\text{H}_{\text{crit}}$. In this case, the Hamiltonian (\ref{eq:HN1}) has, at least for sufficiently large $N$ (so that $\lambda < \lambda_{\text{crit}} (N)$), a unique realization as a self-adjoint operator on $L^2_s (\bR^{3N})$ and therefore generates the one-parameter group of unitary transformation $U_N (t) = e^{-iH_N t}$, with $t\in\bR$. The unique global solution of the $N$-particle Schr\"odinger equation 
\begin{equation}\label{eq:schrN} i\partial_t \psi_{N,t} = H_N \psi_{N,t}, \qquad \psi_{N,t=0} = \psi_N \in L^2_s (\bR^{3N}) \end{equation}
which governs the time evolution of an arbitrary initial $N$-particle wave function $\psi_N$ is then given by $\psi_{N,t} = U_N (t) \psi_N$. 

\medskip

Consider now the time evolution (\ref{eq:schrN}) of a factorized initial data $\psi_N = \ph^{\otimes N}$ for some $\ph \in L^2 (\bR^3)$ (here we use the notation $\ph^{\otimes N} (\bx) = \prod_{j=1}^N \ph (x_j)$, where $\bx = (x_1, \dots ,x_N)$). Of course, factorization is not preserved by the evolution. Nevertheless, because of  the mean field character of the interaction, one may expect factorization of the evolved wave-function $\psi_{N,t}$ to be restored, in an appropriate sense, in the limit of large $N$. If we assume, formally, that \begin{equation}\label{eq:fact-for} \psi_{N,t} \simeq \ph_t^{\otimes N},\end{equation} then it is easy to derive a self-consistent equation for the evolution of the one-particle orbital $\ph_t$. In fact, if (\ref{eq:fact-for}) is correct, the potential experienced by the particles can be approximated by an averaged, mean field, potential given by the convolution $-\lambda \, |.|^{-1} * |\ph_t|^2$. Therefore, (\ref{eq:fact-for}) implies that $\ph_t$ must evolve according to the semirelativistic nonlinear Hartree equation 
\begin{equation}\label{eq:hartree0}
i\partial_t \ph_t =  \sqrt{1-\Delta} \, \ph_t -\lambda \left( \frac{1}{|.|} * |\ph_t|^2 \right) \ph_t 
\end{equation}
with initial data $\ph_{t=0}= \ph$. The question now is in which sense can the factorization (\ref{eq:fact-for}) hold true. It turns out that (\ref{eq:fact-for}) should be understood on the level of the marginal densities associated with $\psi_{N,t}$. 

\medskip

Let $\gamma_{N,t} = |\psi_{N,t} \rangle \langle \psi_{N,t}|$ be the orthogonal projection onto the solution of the $N$-particle Scr\"odinger equation $\psi_{N,t} = e^{-iH_N t} \psi_N$ with factorized initial data $\psi_N = \ph^{\otimes N}$. Then, for $k=1, \dots ,N$, we define the $k$-particle marginal (or reduced) density $\gamma^{(k)}_{N,t}$ associated with $\psi_{N,t}$ by taking the partial trace of $\gamma_{N,t}$ over the last $N-k$ particles; that is 
\[ \gamma^{(k)}_{N,t} = \tr_{k+1, k+2, \dots , N} \, \gamma_{N,t} \, . \]
In other words, $\gamma^{(k)}_{N,t}$ is defined as a non-negative trace class operator on $L^2 (\bR^{3k})$ with kernel 
\[ \begin{split} \gamma^{(k)}_{N,t} (\bx_k, \by_k) &= \int \rd \bx_{N-k} \, \gamma_{N,t} (\bx_k, \bx_{N-k} ; \by_k, \bx_{N-k})  \\ &= \int \rd \bx_{N-k} \, \psi_{N,t} (\bx_k, \bx_{N-k}) \, \overline{\psi}_{N,t} (\by_k, \bx_{N-k}) , \end{split} \] 
where we introduced the notation $\bx_k = (x_1, \dots, x_k)$, $\by_k = (y_1, \dots , y_k)$, $\bx_{N-k} = (x_{k+1}, \dots , x_N)$. 

\medskip

The first rigorous proof of the validity of (\ref{eq:fact-for}) on the level of the marginal densities (and hence the first derivation of (\ref{eq:hartree0})) has been obtained, for the subcritical regime, in \cite{ES}. More precisely, under the condition that $\lambda < 4/\pi$ (which is smaller than $\lambda_{\text{crit}}^\text{H}$), it is proven in \cite{ES} that, for every fixed $k \geq 1$, and for every $t \in \bR$, 
\begin{equation}\label{eq:conv-sub} \tr \, \left| \gamma^{(k)}_{N,t} - |\ph_t \rangle \langle \ph_t|^{\otimes k} \right| \to 0 \end{equation} as $N \to \infty$. Here $\ph_t$ is the solution of the semirelativistic nonlinear Hartree equation (\ref{eq:hartree0}) with initial data $\ph_{t=0} = \ph$. Note that the convergence of the marginal densities implies that, for an arbitrary $k$-particle observable $J^{(k)}$, we have
\[ \langle \psi_{N,t} , (J^{(k)} \otimes 1^{(N-k)} ) \psi_{N,t} \rangle \to \langle \ph_t^{\otimes k} , J^{(k)} \ph_t^{\otimes k} \rangle \] as $N \to \infty$. In this sense, the solution of the $N$-particle Schr\"odinger equation $\psi_{N,t}$ can be approximated, for large $N$, by products of the solution of the one-particle semirelativistic Hartree equation (\ref{eq:hartree0}). Observe here that the semirelativistic Hartree equation (\ref{eq:hartree0}) is locally well-posed in the energy space $H^{1/2} (\bR^3)$, for arbitrary $\lambda \in \bR$. In fact, it is proven in \cite{L} that, for every $\ph \in H^{1/2} (\bR^3)$, there exists a maximal $0 < T \leq \infty$ and a unique solution $\ph_. \in C ((-T,T), H^{1/2} (\bR^3))$ of (\ref{eq:hartree0}) in the time interval $t \in (-T,T)$. Here, either $T=\infty$ (and then the solution is global), or $T <\infty$ and $\| \ph_t \|_{H^{1/2}} \to \infty$ as $t \to T$ or as $t \to -T$ (in this case, $\ph_t$ blows up in finite time). For $\lambda < \lambda^\text{H}_{\text{crit}}$, the unique solution $\ph_t$ is also shown to be global (hence $T=\infty$ in this case); for this reason, (\ref{eq:conv-sub}) makes sense for all $t \in \bR$. 

\medskip

Note that the result (\ref{eq:conv-sub}) is just one of the several results concerning the derivation of effective evolution equations from first principle quantum dynamics which have been obtained in the last years. Other results of this type concern 
the derivation of the non-relativistic Hartree equation in the mean-field limit (see for example \cite{S,EY,BGM,RS,FKP,KP,GMM,GMM2}) and the derivation of the (non-relativistic) Gross-Pitaevskii equation for the description of the dynamics of initially trapped Bose-Einstein condensates (see \cite{ESY1,ESY2,ESY3,ESY4} and, very recently, \cite{P}). Although most of these papers deal with non-relativistic particles, the authors of \cite{KP} also consider semirelativistic bosons interacting through a Newtonian potential, in the subcritical regime $\lambda < \lambda^{\text{H}}_{\text{crit}}$. They improve (\ref{eq:conv-sub}), by giving an explicit bound on the rate of the convergence; more precisely, they show that 
\begin{equation}\label{eq:conv-sub1} \tr \, \left| \gamma^{(k)}_{N,t} - |\ph_t \rangle \langle \ph_t|^{\otimes k} \right| \leq \frac{C_{k,t}}{\sqrt{N}} \end{equation}
for a constant $C_{k,t} =C_t \, \sqrt{k}$, where $C_t$ grows at most exponentially in $t \in \bR$ (under additional assumptions on the dispersion of $\ph_t$, $C_t$ is bounded uniformly in $t$). 

\medskip

So far, we considered the subcritical regime $\lambda < \lambda^{\text{H}}_{\text{crit}}$. Let us discuss now the {\it supercritical} regime $\lambda > \lambda^{\text{H}}_{\text{crit}}$. The criticality of the Hartree energy functional remarked 
earlier can also be observed on the level of the time-dependent semirelativistic nonlinear Hartree equation (\ref{eq:hartree0}). On the one hand, (\ref{eq:hartree0}) is globally well-posed for $\lambda < \lambda^{\text{H}}_{\text{crit}}$. On the other hand, it turns out that, for $\lambda > \lambda^\text{H}_{\text{crit}}$, (\ref{eq:hartree0}) has solutions that blowup in finite time. More precisely, it was proven in \cite{FL} that, for every spherically symmetric $\ph \in H^{1} (\bR^3)$ with $\cE_{\text{Hartree}} (\ph) < 0$ 
such that $\| |x| \ph \| < \infty$, the unique maximal solution $\ph_t \in C((-T,T), H^{1/2} (\bR^3))$ of (\ref{eq:hartree0}) with initial data $\ph_{t=0} = \ph$ blows up in finite time, in the sense that $T<\infty$ and 
\[ \| \ph_t \|_{H^{1/2}} \to \infty, \qquad \text{as } t \to T \quad \text{or as } t \to - T \, . \] 

\medskip

Solutions of (\ref{eq:hartree0}) exhibiting blowup in finite time are supposed to describe, within the framework of Chandrasekhar's theory, the gravitational collapse of bosons stars. The expectation that blowup solutions of the semirelativistic Hartree equation describe the collapse of boson stars is based on the unproven assumption that the many-body dynamics can be approximated by the Hartree dynamics also in the supercritical regime $\lambda > \lambda^\text{H}_{\text{crit}}$ and all the way up to the time of the (nonlinear) blowup. In this paper, we give a rigorous proof of this physical assumption. We show, first of all, that the convergence (\ref{eq:conv-sub1}) also holds in the supercritical case $\lambda \geq \lambda^\text{H}_{\text{crit}}$, if the norm $\| \ph_s \|_{H^{1/2}}$ stays bounded in the interval $[0,t]$. Moreover, we prove that the convergence of the $N$-particle Schr\"odinger evolution towards the Hartree dynamics does not only hold in the sense of (\ref{eq:conv-sub1}); instead, it also holds (again assuming that the norm $\| \ph_s \|_{H^{1/2}}$ remains bounded in $[0,t]$) with respect to the stronger energy norm, at least for the one-particle marginal density (this means that $(1-\Delta)^{1/4} \gamma^{(1)}_{N,t} (1-\Delta)^{1/4}$ converges to $(1-\Delta)^{1/4} |\ph_t \rangle \langle \ph_t| (1-\Delta)^{1/4}$ in the trace norm, as $N \to \infty$). Note that this is the first proof of the convergence of the many body Schr\"odinger evolution towards the Hartree dynamics with respect to the energy norm, not only for supercritical semirelativistic boson stars, but for any mean field system.  As a consequence of the convergence in energy, we show that the solution of the Hartree equation (\ref{eq:hartree0}) really describes the collapse of the many body system.

\medskip

Let us now describe our results in more details. We are interested in the dynamics generated by the $N$-particle Hamiltonian (\ref{eq:HN1}) in the supercritical regime $\lambda > \lambda^\text{H}_{\text{crit}}$ (although at the end our results will also hold in the sub-critical regime). The first issue that we have to face is that, since the form defined by $H_N$ is not bounded below, the Hamiltonian $H_N$ does not necessarily have a unique realization as a self-adjoint operator on the Hilbert space $L_s^2 (\bR^{3N})$ (and thus the time-evolution is not necessarily well-defined). For this reason we choose a sequence $\alpha = (\alpha_N)_{N\geq 1}$ with $\alpha_N > 0$ for all $N \in \bN$ and $\alpha_N \to 0$ as $N \to \infty$, and we define the regularized $N$-particle Hamiltonian 
\begin{equation}\label{eq:regHN} H_N^{\alpha} = \sum_{j=1} \sqrt{1-\Delta_{x_j}} - \frac{\lambda}{N} \sum_{i<j} \frac{1}{|x_i -x_j| + \alpha_N} \,. \end{equation}
The regularized Hamiltonian $H_N^{\alpha}$ defines now a quadratic form on $L^2_s (\bR^{3N})$ which is clearly bounded below, for every $N \in \bN$, since 
\[ \langle \psi_{N} , H_N^{\alpha} \, \psi_{N} \rangle \geq - \frac{\lambda N}{2\alpha_N} \| \psi_N \|^2 \, . \]  By Friedrichs theorem, $H_N^{\alpha}$ has a unique extension as a self-adjoint operator on $L^2_s (\bR^{3N})$, with domain $H^{1/2} (\bR^{3N})$. Hence $H_N^{\alpha}$ generates the one-parameter group of unitary transformations $U_N^\alpha (t) = e^{-itH_N^{\alpha}}$, $t \in \bR$, and therefore the $N$-particles Schr\"odinger equation
\begin{equation}\label{eq:schr-reg} i\partial_t \psi_{N,t} = H^{\alpha}_N \psi_{N,t} \qquad \text{with initial condition } \psi_{N,t=0} = \psi_N \end{equation} is globally well-posed (it has the unique solution $\psi_{N,t} = e^{-iH_N^{\alpha}t} \psi_N$, for all $t \in \bR$). 

\medskip

{F}rom the physical point of view, the introduction of the cutoff $\alpha$ is justified by the observation that on very short length scales, the Newtonian potential is effectively regularized by the presence of other forces (such as electromagnetic or nuclear forces) or because of general relativity effects. The results that we will state and prove below concern the limit of large $N$ and small $\alpha_N$. How fast $\alpha_N$ tends to zero is irrelevant to establish the convergence of the Schr\"odinger evolution to the Hartree dynamics in the trace norm, analogously to (\ref{eq:conv-sub}) (although, of course, the rate of the convergence depends on $\alpha_N$). On the other hand, to show convergence in the energy norm, we will need to assume that $\alpha_N$ does not converge to zero too fast; more precisely, we will suppose that there exists $\beta >0$ such that $N^{\beta} \alpha_N \to \infty$. This condition, which allow for any power law decay, still leaves a lot of 
freedom in the choice of $\alpha_N$ (physically, conditions on the decay of $\alpha$ translate into restrictions of the range of systems for which the approximation of the many body evolution by the Hartree dynamics is applicable). 

\medskip

We study the time evolution generated by the regularized Hamiltonian (\ref{eq:regHN}) on factorized initial data $\psi_N = \ph^{\otimes N}$ for $\ph \in H^2 (\bR^3)$. We compare the marginal densities associated with the solution of the $N$ particle Schr\"odinger equation $\psi_{N,t} = e^{-iH_N^{\alpha} t} \psi_N$ with products of the solution to the Hartree equation (\ref{eq:hartree0}). Note that the cutoff disappears in the limiting Hartree equation, because of the assumption that $\alpha_N \to 0$ as $N \to \infty$ (part of the proof of the convergence will consists in estimating the distance between the solution $\ph_t$ of (\ref{eq:hartree0}) and the solution $\ph_t^{(\alpha)}$ of a regularized Hartree equation with interaction $-\lambda / |x|$ replaced by the regularized interaction $-\lambda/ (|x| + \alpha)$ in the limit of small $\alpha$). 

\medskip

The first main result of this paper is the following theorem. Under the assumption that the solution $\ph_t$ of (\ref{eq:hartree0}) has a bounded $H^{1/2}$-norm in the interval $[-T,T]$ (which means that there is no blowup, up to time $T$), we prove the convergence of the marginal densities associated with the solution $\psi_{N,t}$ of the $N$-particle Schr\"odinger equation (\ref{eq:schr-reg}) to the orthogonal projections onto products of $\ph_t$. The theorem also gives an explicit upper bound on the fluctuations around the Hartree dynamics.

\begin{theorem}\label{Thm:no-blowup}
Fix $\lambda \in \bR$, $\ph \in H^2 (\bR^3)$ with $\| \ph \| = 1$ and set $\psi_N = \ph^{\otimes N}$. 
Let $\psi_{N,t} = e^{-iH^{\alpha}_N t} \psi_N$ be the evolution of the the initial wave function $\psi_N$ with respect to the Hamiltonian (\ref{eq:regHN}), and let $\gamma_{N,t}^{(1)}$ be the one-particle reduced density associated with $\psi_{N,t}$. 

\medskip

Denote by $\ph_t$ the solution of the nonlinear Hartree equation (\ref{eq:hartree0}) 
with initial data $\ph_{t=0} = \ph$. Fix $T>0$ such that
\begin{equation}\label{eq:ka1}
\kappa := \sup_{|t|\leq T} \|\ph_t\|_{H^{1/2}} <  \infty \,.
\end{equation}

\medskip

Then there exists a constant $C=C(\kappa,T, \| \ph \|_{H^2})$ such that
\begin{equation}\label{eq:conv1-tr}
\trr\,\big|\gamma_{N,t}^{(1)}-|\ph_t\rangle\langle\ph_t|\,\big|\;\leq\;C \, \Big(\frac{1}{\sqrt{N}}+\alpha_N\Big) \,.
\end{equation}
for all $t\in \bR$ with $|t| \leq T$, and for all $N$ sufficiently large. In particular, if $\alpha_N \to 0$ as $N \to \infty$, it follows that $\gamma_{N,t}^{(1)} \to |\ph_t \rangle \langle \ph_t|$ in trace-norm, as $N \to \infty$.
\end{theorem}

{\it Remarks.} The existence of $T>0$ such that (\ref{eq:ka1}) is satisfied is a consequence of the local well-posedness of the nonlinear Hartree equation (\ref{eq:hartree0}), see \cite{L}. Similar methods to the ones used to prove (\ref{eq:conv1-tr}) can be employed to show the convergence of 
higher order marginals $\gamma^{(k)}_{N,t}$ with the same rate $(N^{-1/2} + \alpha_N)$ for any fixed $k \in \bN$. If we are satisfied with a slower rate for higher marginals, a simple argument, outlined in Section 2 of \cite{KP}, shows that (\ref{eq:conv1-tr}) immediately implies that for any $k \in \bN$, 
\[ \trr\,\big|\gamma_{N,t}^{(k)}-|\ph_t\rangle\langle\ph_t|^{\otimes k} \,\big|\;\leq\;C \, \sqrt{k\Big(\frac{1}{\sqrt{N}}+\alpha_N\Big)} \,. \]

\medskip

The first ingredient of the proof of Theorem \ref{Thm:no-blowup} is the observation that the bound (\ref{eq:ka1}) on the $H^{1/2}$-norm of the solution of the Hartree equation (\ref{eq:hartree0}), together with the assumption $\ph \in H^2 (\bR^3)$ on the initial data, implies an upper bound on the $H^{1/2}$-norm of the solution of the regularized Hartree equation 
\begin{equation}\label{eq:reg-har1}
i\partial_t \varphi^{(\alpha)}_t=\sqrt{1-\Delta}\varphi^{(\alpha)}_t-\lambda\Big(\frac{1}{|\cdot|+\alpha}*|\varphi^{(\alpha)}_t|^2\Big)\varphi^{(\alpha)}_t \, ,
\end{equation}
uniform in the cutoff $\alpha>0$ (actually, we prove that $\| \ph_t - \ph_t^{(\alpha)}\|_{H^{1/2}}$ is small, of the order $\alpha^{1/2}$). By the propagation of regularity for the solution of the Hartree equation (both the original equation (\ref{eq:hartree0}) and the regularized equation (\ref{eq:reg-har1})), we also obtain a bound for the norm $\| \ph^{(\alpha)}_t \|_{H^2}$ uniform in $\alpha >0$ and in $t \in [-T,T]$. 

\medskip

The second ingredient in the proof of Theorem \ref{Thm:no-blowup} is the method developed in \cite{RS} to establish the convergence to the Hartree dynamics for a system of non-relativistic bosons. This method is based on the use of a Fock space representation of the many boson system, and on the study of the dynamics of coherent states. When analyzing the time-evolution of an initial coherent state, it is possible to isolate the Hartree component of the evolution. Moreover, as first observed in \cite{H} (and later in \cite{GV}), the evolution of the fluctuations around the Hartree dynamics can be expressed through a two-parameter group of unitary evolutions with an explicit time-dependent generator. The problem then reduces to deriving a bound for the growth of the number of particle operator (which measures the ``number'' of fluctuations, after second quantization) with respect to this evolution. The crucial observation is that the bound derived in \cite{RS} for non-relativistic particles can be easily extended to the relativistic setting, once a uniform bound on $\| \ph_t^{(\alpha)} \|_{H^1}$ is available. 

\medskip

In Section \ref{sec:NLS}, we show the necessary bounds on the solution $\ph^{(\alpha)}_t$ of the regularized equation (\ref{eq:reg-har1}). In Section \ref{sec:Fock} we introduce the Fock space representation, we define  the coherent states, and we discuss some of their main properties. Then, in Section \ref{sec:coh-trace}, we show that the evolution of initial coherent states can be approximated by the Hartree dynamics, and we use this fact to conclude the proof of Theorem \ref{Thm:no-blowup}.

\medskip

The next theorem is the second main result of this paper. It establishes the convergence of the one-particle marginal density associated with the solution of the regularized $N$-particle Schr\"odinger equation to the orthogonal projection onto the solution of the Hartree equation (\ref{eq:hartree0}), in the energy norm. The result holds for all $t \in [-T,T]$ under the assumption that the $H^{1/2}$ norm of the solution $\ph_t$ of (\ref{eq:hartree0}) remains bounded in $[-T,T]$ (in other words, under the assumption that there is no blowup, up to time $T$).

\begin{theorem}\label{thm:conv-en}
Fix $\ph \in H^2 (\bR^3)$ with $\| \ph \| = 1$ and set $\psi_N = \ph^{\otimes N}$.  Consider an arbitrary sequence $\alpha_N > 0$ with $\alpha_N \to 0$ and such that $N^{\beta} \alpha_N \to \infty$ as $N \to \infty$, for some $\beta >0$. Let $\psi_{N,t} = e^{-iH^{\alpha}_N t} \psi_N$ be the evolution of the initial wave function $\psi_N$ generated by the Hamiltonian (\ref{eq:regHN}) and let $\gamma_{N,t}^{(1)}$ be the one-particle reduced density associated with $\psi_{N,t}$. 

\medskip

Denote by $\ph_t$ the solution of the nonlinear Hartree equation (\ref{eq:hartree0}) with initial data $\ph_{t=0} = \ph$. Fix $T>0$ such that
\begin{equation}\label{eq:ka}
\kappa := \sup_{|t|\leq T} \|\ph_t\|_{H^{1/2}} <  \infty \,.
\end{equation}

\medskip

Then there exists a constant $C=C(\kappa,T, \| \ph \|_{H^2},\beta)$ such that
\begin{equation}\label{eq:conv-ene}
\trr\,\left|(1-\Delta)^{1/4} \left(\gamma_{N,t}^{(1)}-|\ph_t\rangle\langle\ph_t|\right) (1-\Delta)^{1/4} \,\right|\;\leq\;C \, \Big(\frac{1}{N^{1/4}}+ \alpha^{1/2}_N \Big)
\end{equation}
for all $t\in \bR$ with $|t| \leq T$, and for all $N$ sufficiently large. In particular, it follows that $\gamma_{N,t}^{(1)} \to |\ph_t \rangle \langle \ph_t|$ in energy-norm, as $N \to \infty$.
\end{theorem}

{\it Remark.} We believe that the same arguments used to show (\ref{eq:conv-ene}) can be extended to prove the convergence (in energy norm) of the higher marginal $\gamma^{(k)}_{N,t}$. To keep the paper readable, we do not follow this direction here. 
Note that a simple argument, similar to the one presented in Section 2 of \cite{KP} (and mentioned in the remark after Theorem \ref{Thm:no-blowup}), shows that (\ref{eq:conv-ene}) implies 
\[ \trr\,\left|(1-\Delta_{x_1})^{1/4} \left(\gamma_{N,t}^{(k)}-|\ph_t\rangle\langle\ph_t|^{\otimes k} \right) (1-\Delta_{x_1})^{1/4} \,\right|\;\leq\;C \,\sqrt{k \, \Big(\frac{1}{N^{1/4}}+ \alpha^{1/2}_N \Big)} \, 
\] for all $k\in \bN$. 

\bigskip
 
The proof of Theorem \ref{thm:conv-en} is based again on a Fock space representation of the many body system, and on the use of coherent states as initial data. As in the proof of Theorem \ref{Thm:no-blowup}, the Hartree component of the evolution of the initial coherent states can be isolated, and the dynamics of the fluctuations can be written through a two-parameter group of unitary transformations with an explicit generator. To obtain convergence in the energy norm, however, instead of controlling the growth of the number of particle operator, we need to control the growth of the kinetic energy operator with respect to the fluctuation dynamics. Technically, this step (contained in Proposition \ref{prop:Kbd}) is the most challenging part of our paper. In a sense, the fact that we can control the growth of the kinetic energy of the fluctuations implies that, although on the $N$ particle level we are considering a supercritical regime, after subtracting the (supercritical) Hartree dynamics, the system on the level of the fluctuations is subcritical. 
In Section \ref{sec:conv-en} we prove the convergence to the Hartree dynamics for initial coherent states and we complete the proof of Theorem \ref{thm:conv-en} assuming Proposition \ref{prop:Kbd} to hold true. In Section \ref{sec:Kbd}, we prove Proposition \ref{prop:Kbd}.

\medskip

Theorems \ref{Thm:no-blowup} and \ref{thm:conv-en} show that, as long as a bound on the $H^{1/2}$-norm of the solution $\ph_t$ of the Hartree equation (\ref{eq:hartree0}) is available, the evolution of the marginal densities can still be approximated by $\ph_t$. Next, we ask what happens if the solution $\ph_t$ of (\ref{eq:hartree0}) exhibits blowup. Under the assumption that $\ph_t$ blows up as $t \to T$, for some $0< T<\infty$, we show that also the solution of the regularized $N$-particle Schr\"odinger equation (\ref{eq:schr-reg}) collapses, if $t \to T$, and, simultaneously, $N \to \infty$. The $N$-particle wave function $\psi_{N,t}$ collapses in the sense that the kinetic energy per particle, which remains finite, uniformly in $N$, up to time $T$, diverges to infinity as $t \to T$ if simultaneously, $N \to \infty$.  In order to make sure that the solution of the $N$ particle Schr\"odinger equation remains close to the solution of the Hartree equation as $t$ approaches the nonlinear blowup time, we have to assume that $N$ diverges to infinity sufficiently fast. Physically, this condition imposes restrictions to the range of many body systems for which the Hartree approximation is valid close to the blowup time. {F}rom a different point of view (if we think of the number of particles $N$ as fixed), it tells us how close to the nonlinear blowup time we can expect the Hartree dynamics to be a good approximation for the real many body quantum evolution. 

\begin{corollary}\label{Thm:with-blowup}
Fix $\ph \in H^2 (\bR^3)$ with $\| \ph \| = 1$ and set $\psi_N = \ph^{\otimes N}$. Consider an arbitrary sequence $\alpha_N >0$ with $\alpha_N \to 0$ and $N^{\beta} \alpha_N \to \infty$ as $N \to \infty$, for some $\beta >0$.  Let $\psi_{N,t} = e^{-iH^{\alpha}_N t} \psi_N$ be the evolution of the initial wave function $\psi_N$ generated by the Hamiltonian (\ref{eq:regHN}) and let $\gamma_{N,t}^{(1)}$ be the one-particle reduced density associated with $\psi_{N,t}$. 

\medskip

Denote by $\ph_t$ the solution of the nonlinear Hartree equation (\ref{eq:hartree0})
with initial data $\ph_{t=0} = \ph$. Suppose that $T_c >0$ is the first time of blow-up for $\ph_t$. In other words, assume that
\[ \kappa_t :=  \sup_{0< s < t} \|\ph_s\|_{H^{1/2}} < \infty \] for all $t<T_c$, and 
\[ \| \ph_t \|_{H^{1/2}} \to \infty \quad \text{as } t \to T_c^- \, . \]

\medskip

Then, for any fixed $t \in [0,T_c)$ there exists a constant $C_t >0$ such that
\[ \| (1-\Delta_{x_1})^{1/4} \, \psi_{N,t} \|^2 \trr \, (1-\Delta)^{1/2}\gamma_{N,t}^{(1)} < C_t \]
uniformly in $N \in \bN$. Moreover, for $t \in [0,T_c)$, there exists $N (t) \in \bN$ with $N(t) \to \infty$ as $t \to T_c^-$, and such that 
\begin{equation}
\| (1-\Delta_{x_1})^{1/4} \, \psi_{N(t),t} \|^2 = \trr \, (1-\Delta)^{1/2}\gamma_{N(t),t}^{(1)}\to \infty\qquad\mathrm{as}\quad t\to T^-_c\, .
\end{equation}
In other words, the kinetic energy per particle is uniformly bounded in $N$, if $0 \leq t < T_c$ but it diverges in the limit $t \to T_c^-$, if at the same time, the number of particles tends to infinity sufficiently fast. 
\end{corollary}

{\it Remark 1.} The existence of blow-up for solutions of the nonlinear Hartree equation (\ref{eq:hartree0}) has been proven in \cite{FL} under the assumption that the initial data $\ph$ is spherically symmetric and that it has negative energy $\cE_{\text{Hartree}} (\ph) < 0$; see (\ref{eq:enhar}) (this is possible if $\lambda > \lambda_{\text{crit}}^\text{H}$). 

\medskip

{\it Remark 2.} The fact that $\psi_{N,t}$ collapses at some point in the interval $[0,T_c]$ follows already from the blow-up of $\ph_t$ at time $T_c$ and from the Theorem~\ref{Thm:no-blowup}. This fact, which was pointed out to us by R. Seiringer, follows from the general observation that the kinetic energy of an $L^2$-limit is always smaller than the limit of the kinetic energy. This argument, however, does not prove that the collapse takes place at time $T_c$ nor that the blow-up of the Hartree equation accurately describes it. 

\medskip

\begin{proof}[Proof of Corollary \ref{Thm:with-blowup}]
Set $\beta_N = N^{-1/4} + \alpha_N$; then $\beta_N \to 0$ as $N \to \infty$. 
For every $0< t < T_c$ there exists, by Theorem 1.2, a constant $C_t$, depending on $\beta, \kappa_t,t,\| \ph \|_{H^2}$, such that
\[ \trr\,\left|(1-\Delta)^{1/4} \left(\gamma_{N,s}^{(1)}-|\ph_s\rangle\langle\ph_s|\right) (1-\Delta)^{1/4} \,\right|\;\leq\;C_t \beta_N \]
for all $0<s<t$. In particular this implies that
\[ \left| \tr \, (1-\Delta)^{1/2} \gamma_{N,t}^{(1)} - \| \ph_t \|^2_{H^{1/2}} \right| \leq C_t \beta_N \, .\]
For $0<t < T_c$, choose now $N(t)$ sufficiently large , so that $\gamma_{N(t)} \leq (T_c - t) /C_t$ (this is certainly possible because $\beta_N \to 0$ as $N \to \infty$).  Then 
\[ \left| \tr \, (1-\Delta)^{1/2} \gamma_{N(t),t}^{(1)} - \| \ph_t \|^2_{H^{1/2}} \right| \to 0 \] as $t \to T_c$. Since $\| \ph_t \|_{H^{1/2}} \to \infty$ as $t \to T_c$, this implies that
\[ \tr\, (1-\Delta)^{1/2} \gamma_{N(t),t}^{(1)} \to \infty \, . \]
\end{proof}

{\it Acknowledgements.} B. Schlein would like to thank I. Rodnianski for precious discussions, from which the present paper originated. He is also happy to 
thank C. Hainzl for reading a previous version of this work and for clarifying to him the relation between the critical constants $\lambda_{\text{crit}} (N)$ and $\lambda^{\text{H}}_{\text{crit}}$. Moreover he would like to thank M. Lewin and R. Seiringer for useful discussions and remarks.  A. Michelangeli was partially supported by a INdAM-GNFN grant "Progetto Giovani 2009". B. Schlein acknowledges support from the ERC Starting Grant MAQD - 240518.

\section{Bounds on solutions of nonlinear Hartree equations}
\label{sec:NLS}

In this section, we study properties of the solution of the nonlinear Hartree equation (\ref{eq:hartree0}). In particular, we need to compare the solution of (\ref{eq:hartree0}) with the solution of regularized Hartree equations (like (\ref{NLS_Hartree_rel}), with $\alpha>0$). To this end, we first need to establish the property of propagation of initial regularity, under the assumption of a bound on the $H^{1/2}$-norm.

\begin{proposition}[Propagation of regularity] \label{Prop:propagation_of_regularity}
Fix $s > 1/2$ and $\alpha \geq 0$. Let $\ph \in H^s(\mathbb{R}^3)$ with $\|\ph\|=1$. Let $\ph_t$ denote the solution to the nonlinear Hartree equation 
\begin{equation}\label{NLS_Hartree_rel}
i\partial_t\varphi_t=\sqrt{1-\Delta}\varphi_t-\lambda\Big(\frac{1}{|\cdot|+\alpha}*|\varphi_t|^2\Big)\varphi_t
\end{equation}
with the initial condition $\ph_{t=0}=\ph$. Fix $T >0$ such that 
\begin{equation}\label{assumption_boundednessH12}
\kappa := \sup_{|t|\leq T}\|\varphi_t\|_{H^{1/2}} \;<\; \infty\,.
\end{equation}
Then there exists a constant $\nu = \nu (\kappa,T,s, \| \ph \|_{H^s}) < \infty$ (but independent of $\alpha$) such that
\begin{equation}\label{phi_Hs}
\sup_{|t|\leq T}\|\varphi_t\|_{H^s} \leq \nu \, .  
\end{equation}
\end{proposition}

\bi

\begin{proof}
We follow here the proof of \cite[Lemma 3]{L} with some modifications. Let $J(\ph):= \big( (|\,\cdot\,|+\alpha)^{-1}*|\ph|^2\big) \ph$. Then we claim that 
\begin{equation}\label{Lenzmann-estimate}
\|J(\ph)\|_{H^s}\;\lesssim\;\|\ph\|_{H^{1/2}}^2\,\|\ph\|_{H^s}
\,,\qquad \text{ for all } \ph \in H^s(\mathbb{R}^3)\,.
\end{equation} 
In fact, $\|J(\ph)\|_{H^s}\lesssim\|J(\ph)\|_{2}+\|(-\Delta)^{s/2}J(\ph)\|_{2}$ and
\begin{equation}\label{L__2}
\|J(\ph)\|_{2}\;\lesssim\;\Big\|\Big(\frac{1}{|\cdot|+\alpha}*|\ph|^2\Big)\ph \Big\|_2\;\lesssim\;\Big\|\frac{1}{|\cdot|}*|\ph|^2\Big\|_{\infty}\,\|\ph\|_2\;\lesssim\;\|\ph\|_{H^{1/2}}^2\|\ph\|_2\,.
\end{equation} 
Moreover,
\begin{equation}\label{L__>1}
\begin{split}
\big\|(-\Delta)^{s/2}J(\psi)\big\|_{2}\;&=\;\left\|(-\Delta)^{s/2} \left(\frac{1}{|.| +\alpha} * |\ph|^2\right) \ph \right\|_2 \\
& \lesssim \; \left\| \frac{1}{|.|} * | (-\Delta)^{s/2}|\ph|^2 | \right\|_{6}  \, \|\ph\|_{3} +\left\|\frac{1}{|.| + \alpha} * |\ph|^2 \right\|_{\infty} \, \big\|(-\Delta)^{s/2}\ph \big\|_2 \\
& \lesssim \; \big\|(-\Delta)^{s/2}|\ph|^2 \|_{6/5} \| \ph \|_{H^{1/2}} + \| \ph \|^2_{H^{1/2}} \| \ph \|_{H^s} \\ & \lesssim \| \ph \|^2_{H^{1/2}} \| \ph \|_{H^s} \,. 
\end{split}
\end{equation}
Here we used the generalized Leibniz rule (see Lemma \ref{lm:leib}) in the first inequality. In the second inequality, we used the Hardy-Littlewood-Sobolev inequality and the Sobolev inequality $\| \ph \|_3 \lesssim \| \ph \|_{H^{1/2}}$ to bound the first term, and Kato's inequality 
\begin{equation}\label{eq:kato}
\sup_{x \in \bR^3} \int \rd y \, \frac{|\ph (y)|^2}{|x-y|}  \leq \frac{\pi}{2} \int \rd y \, \left| |\nabla|^{1/2} \, \ph (y) \right|^2 \leq  \frac{\pi}{2}  \, \| \ph \|^2_{H^{1/2}} 
\end{equation}
to bound the second term. Finally, in the third inequality, we used again the generalized Leibniz rule. This shows (\ref{Lenzmann-estimate}). 

\medskip

Next, we write $\ph_t$ as 
\begin{equation}
\varphi_t\;=\;e^{-i\sqrt{1-\Delta}\,t}\varphi+i\lambda\int_0^t\rd s\, e^{-i\sqrt{1-\Delta}\,(t-s)}\Big(\frac{1}{|.|+\alpha_N}*|\varphi_s|^2\Big)\,\varphi_s
\end{equation}
and we obtain, by \eqref{Lenzmann-estimate} and \eqref{assumption_boundednessH12}, that
\begin{equation}\label{Hs_pre_gronwall}
\begin{split}
\|\varphi_t\|_{H^s}\;&\lesssim\;\|\varphi\|_{H^s}+\int_0^t\|J(\varphi_\tau)\|_{H^s}\,\rd \tau\;\lesssim\;\| \ph \|_{H^s} +\kappa^2\int_0^t\|\varphi_\tau\|_{H^s}\,\rd \tau\,.
\end{split}
\end{equation} 
The proposition now follows applying Gronwall's inequality to \eqref{Hs_pre_gronwall}.
\end{proof}

Next, under the assumption that  the solution $\ph_t$ of the Hartree equation (\ref{eq:hartree0}) with initial data $\ph \in H^2 (\bR^3)$ stays bounded in $H^{1/2}$ in the interval $[-T,T]$, we show the vicinity (in the $H^{1/2}$-norm) of the solution $\ph_t^{(\alpha)}$ of the regularized equation (\ref{NLS_Hartree_rel}), with $\alpha >0$ and small, to $\ph_t$. 

\begin{proposition}\label{closeness_in_H1/2}
Fix $\ph \in H^1(\mathbb{R}^3)$ with $\| \ph \| = 1$ and let $\ph_t$ denote the solution of the nonlinear Hartree equation (\ref{eq:hartree0}). 
with initial condition $\varphi_{t=0}=\varphi$. Let $T>0$ be such that
\begin{equation}\label{eq:kaa}
\kappa := \sup_{|t|\leq T}\|\varphi_t\|_{H^{1/2}} \;< \;\infty\,.
\end{equation}
For $\alpha>0$, let $\varphi_t^{(\alpha)}$ be the solution to the regularised Hartree  equation
\begin{equation}\label{eq:hara1}
i\partial_t\varphi_t^{(\alpha)}=\sqrt{1-\Delta}\, \varphi_t^{(\alpha)}-\lambda\Big(\frac{1}{|\cdot|+\alpha}*|\varphi_t^{(\alpha)}|^2\Big)\varphi_t^{(\alpha)}
\end{equation}
with initial condition $\varphi_{t=0}=\varphi$. 

\medskip

Then there exists a constant $C = C (T,\kappa, \| \ph \|_{H^1})<\infty$, such that
\begin{equation}\label{distance_in_L2}
\big\|\varphi_t-\varphi_t^{(\alpha)}\big\|_2\;\leq\;C \,\alpha \, \qquad \text{for all $|t|\leq T$ and all } \alpha >0.
\end{equation} 

\medskip

Moreover, if we assume additionally that $\varphi\in H^2(\mathbb{R}^3)$, then we can also find a constant $D = D (T,\kappa, \| \ph \|_{H^2}) < \infty$ such that 
\begin{equation}\label{H1/2closeness}
\big\|\varphi_t-\varphi_t^{(\alpha)}\big\|_{H^{1/2}}\;\leq\; D \,\alpha^{1/2}\, ,
\end{equation}
for all $|t| \leq T$ and $0 < \alpha <1$. 
\end{proposition}

\bigskip

\begin{proof}
We start by proving (\ref{distance_in_L2}). Since $\varphi\in H^1(\mathbb{R}^3)$, 
we can find, by (\ref{eq:kaa}) and Proposition \ref{Prop:propagation_of_regularity}, $\nu = \nu (T,\kappa, \|\ph \|_{H^1}) < \infty$ such that
\begin{equation}
\sup_{|t|\leq T}\|\varphi_t\|_{H^1}\; \leq \;\nu  \, .
\end{equation}

Let $t\in[-T,T]$. We have
\begin{equation*}
\begin{split}
\frac{\rd}{\rd t}\|\varphi_t-\varphi_{t}^{(\alpha)}\|_2^2 \; = &\; -2\frac{\rd}{\rd t} \re\la \varphi_{t},\varphi_{t}^{(\alpha)}\ra \\ = &\;2\lambda\,\im\Big\la \varphi_t\,,\Big(\frac{1}{|x|}*|\varphi_t|^2-\frac{1}{|x|+\alpha}*|\varphi_{t}^{(\alpha)}|^2 \Big)\varphi_{t}^{(\alpha)}\Big\ra \\
 = &\; 2\lambda\,\im\bigg\{\Big\la \varphi_t\,,\Big(\frac{\alpha}{|x|(|x|+\alpha))}*|\varphi_t|^2\Big)(\varphi_{t}^{(\alpha)}-\varphi_t)\Big\ra + \\ &\hspace{1cm} +
\Big\la \varphi_t\,,\Big(\frac{1}{|x|+\alpha}*\big(|\varphi_t|^2-|\varphi_{t}^{(\alpha)} |^2\big)\Big)(\varphi_{t}^{(\alpha)}-\varphi_t)\Big\ra\,\bigg\}\,. 
\end{split}
\end{equation*}
Therefore
\begin{equation}\label{Gronwall-L2-start}
\begin{split}
\bigg|\, \frac{\rd}{\rd t}\|\varphi_t-\varphi_{t}^{(\alpha)}\|_2^2\,\bigg|\;& \leq\;2|\lambda|\,\bigg\{\,\Big|\, \la \varphi_t\,,\Big(\frac{\alpha}{|x|(|x|+\alpha))}*|\varphi_t|^2\Big)(\varphi_{t}^{(\alpha)}-\varphi_t)\Big\ra \,\Big| \\
& \qquad + \Big|\, \la \varphi_t\,,\Big(\frac{1}{|x|+\alpha}*\big(|\varphi_t|^2-|\varphi_{t}^{(\alpha)} |^2\big)\Big)(\varphi_{t}^{(\alpha)}-\varphi_t)\Big\ra  \,\Big|\,\bigg\}\,.
\end{split}
\end{equation}
The first summand in the r.h.s.~of \eqref{Gronwall-L2-start} can be estimated as
\begin{equation*}
\begin{split}
\Big|\, \la \varphi_t\,,\Big(\frac{\alpha}{|x|(|x|+\alpha))}*|\varphi_t|^2\Big)(\varphi_{t}^{(\alpha)}-\varphi_t)\Big\ra \,\Big| \; & \leq \; \Big\|\,\frac{\alpha}{|x|(|x|+\alpha))}*|\varphi_t|^2\,\Big\|_{\infty}\|\varphi_t-\varphi_t^{(\alpha)} \|_2 \\
& \leq\;\alpha\,\Big\|\,\frac{1}{\,|x|^2}*|\varphi_t|^2\,\Big\|_{\infty} \|\varphi_t-\varphi_t^{(\alpha)} \|_2 \\
&\lesssim \; \alpha\, \|\varphi_t-\varphi_t^{(\alpha)} \|_2 \|\varphi_t\|_{H^1}^2\;\lesssim\;\alpha\,\nu^2\index{\footnote{}} \|\varphi_t-\varphi_t^{(\alpha)} \|_2\,.
\end{split}
\end{equation*}
The second summand can be estimated by
\begin{equation*}
\begin{split}
& \Big|\, \la \varphi_t\,,\Big(\frac{1}{|x|+\alpha}*\big(|\varphi_t|^2-|\varphi_{t}^{(\alpha)} |^2\big)\Big)(\varphi_{t}^{(\alpha)}-\varphi_t)\Big\ra  \,\Big| \\
& \leq \; \int_{\mathbb{R}^3\times\mathbb{R}^3}\rd x \,\rd y\,|\varphi_t(x)| \,\frac{1}{|x-y|+\alpha}\,|\varphi_t(y)-\varphi_{t}^{(\alpha)}(y)|\,\big(|\varphi_{t}(y)|+|\varphi_{t}^{(\alpha)}(y)|\big)\,\big|\varphi_{t}^{(\alpha)}(x)-\varphi_t(x)\big|\, \\
& \leq\; \int_{\mathbb{R}^3\times\mathbb{R}^3}\rd x \,\rd y\,|\varphi_t(x)|^2\,\frac{1}{\,|x-y|^2}|\varphi_t(y)-\varphi_{t}^{(\alpha)}(y)|^2 \\
& \qquad + 2\int_{\mathbb{R}^3\times\mathbb{R}^3}\rd x \,\rd y\,\big(|\varphi_{t}(y)|^2+|\varphi_{t}^{(\alpha)}(y)|^2\big)\,|\varphi_{t}^{(\alpha)}(x)-\varphi_t(x)|^2 \\
& \leq \; \Big\|\frac{1}{|x|^2}*|\varphi_t|^2\Big\|_\infty\,\|\varphi_t-\varphi_{t}^{(\alpha)}\|_2^2\;+\;2\,\big(\|\varphi_t\|_2^2+\|\varphi_t^{(\alpha)}\|_2^2\big)\,\|\varphi_t-\varphi_{t}^{(\alpha)}\|_2^2 \\
& \lesssim \; (\|\varphi_t\|_{H^{1}}^2 + 1) \,\|\varphi_t-\varphi_{t}^{(\alpha)}\|_2^2\;\lesssim\;(1+\nu^2) \, \|\varphi_t-\varphi_{t}^{(\alpha)}\|_2^2\, ,
\end{split}
\end{equation*}
where on the last line we used Hardy's inequality 
\begin{equation}\label{eq:hardy}
\sup_{x\in \bR^3} \int \rd y \, \frac{|\ph (y)|^2}{|x-y|^2} \leq 4 \int \rd y \, |\nabla \ph (y)|^2 \leq 4 \| \ph \|^2_{H^1} \, . \end{equation}
Thus, \eqref{Gronwall-L2-start} gives
\begin{equation}\begin{split}
\frac{\rd}{\rd t}\|\varphi_t-\varphi_{t}^{(\alpha)}\|_2^2 \; & \lesssim \; (1+\nu^2) \,\Big(\alpha\,\|\varphi_t-\varphi_t^{(\alpha)} \|_2+ \|\varphi_t-\varphi_{t}^{(\alpha)}\|_2^2\Big) \\
& \lesssim\; (1+ \nu^2) \,\Big(\alpha^2+ \|\varphi_t-\varphi_{t}^{(\alpha)}\|_2^2\Big)\,.
\end{split}
\end{equation}
By Gronwall's inequality, we find $C= C (\nu,T)$ with 
\begin{equation}
\|\varphi_t-\varphi_{t}^{(\alpha)}\|_2\;\leq\;C  \alpha \, 
\end{equation}
for all $\alpha >0$.

\bi

Next, we prove (\ref{H1/2closeness}). To this end, it is enough to show that there exists $D=D(T,\kappa,\| \ph \|_{H^2})$ such that 
\begin{equation}\label{Delta^1/4_distance_in_L2}
 \big\|(-\Delta)^{1/4}\big(\varphi_t-\varphi_{t}^{(\alpha)}\big)\big\|_2\;\lesssim\; D \, \alpha^{1/2} \, .
\end{equation}
Note that, since $\varphi\in H^2 (\mathbb{R}^3)$, 
we can find, by (\ref{eq:kaa}) and Proposition \ref{Prop:propagation_of_regularity}, $\nu = \nu (T,\kappa, \|\ph \|_{H^2}) < \infty$ such that
\begin{equation}\label{eq:H2-bd}
\sup_{|t|\leq T}\|\varphi_t\|_{H^2}\; \leq \;\nu  \, .
\end{equation}


\bi

We write $\ph_t$ and $\ph_t^{(\alpha)}$ using their Duhamel expansions 
\begin{equation*}
\varphi_t\;=\;e^{-i\sqrt{1-\Delta}\,t}\varphi+i\lambda\int_0^t\rd s\, e^{-i\sqrt{1-\Delta}\,(t-s)}\Big(\frac{1}{|x|}*|\varphi_s|^2\Big)\,\varphi_s
\end{equation*}
and
\begin{equation*}
\varphi_t^{(\alpha)}\;=\;e^{-i\sqrt{1-\Delta}\,t}\varphi+i\lambda\int_0^t\rd s\, e^{-i\sqrt{1-\Delta}\,(t-s)}\Big(\frac{1}{|x|+\alpha}*|\varphi_s^{(\alpha)}|^2\Big)\,\varphi_s^{(\alpha)}
\end{equation*}
respectively. Thus
\begin{equation}\begin{split}
\big\|(-\Delta)^{1/4}\big(\varphi_t-\varphi_{t}^{(\alpha)}\big)\big\|_2\;\leq\;|\lambda|\int_0^t\rd s\,\bigg\{\,&\Big\|(-\Delta)^{1/4}\Big(\frac{1}{|x|}*|\varphi_s|^2\Big)(\varphi_s-\varphi_s^{(\alpha)}) \Big\|_2 \\
& +\Big\|(-\Delta)^{1/4}\Big(\frac{\alpha}{|x|(|x|+\alpha)}*|\varphi_s|^2\Big)\varphi^{(\alpha)}_s \,\Big\|_2 \\
& +\Big\|(-\Delta)^{1/4}\Big(\frac{1}{|x|+\alpha}*\big(|\varphi_s|^2-|\varphi_s^{(\alpha)}|^2\big)\Big) \varphi_s^{(\alpha)} \Big\|_2 \,\bigg\}\,.
\end{split}
\end{equation}
Further decomposing the second and third term in the parenthesis we find
\begin{equation}\label{D^1/4-splitting}\begin{split}
\big\|(-\Delta)^{1/4}\big(\varphi_t-\varphi_{t}^{(\alpha)}\big)\big\|_2 \leq\;|\lambda|\int_0^t\rd s\,\bigg\{\,&\Big\|(-\Delta)^{1/4}\Big(\frac{1}{|x|}*|\varphi_s|^2\Big)(\varphi_s-\varphi_s^{(\alpha)}) \Big\|_2 \\
& +\Big\|(-\Delta)^{1/4}\Big(\frac{\alpha}{|x|(|x|+\alpha)}*|\varphi_s|^2\Big)(\varphi_s-\varphi_s^{(\alpha)}) \Big\|_2 \\
& +\Big\|(-\Delta)^{1/4}\Big(\frac{\alpha}{|x|(|x|+\alpha)}*|\varphi_s|^2\Big)\,\varphi_s \,\Big\|_2 \\
& +\Big\|(-\Delta)^{1/4}\Big(\frac{1}{|x|+\alpha}*\big(|\varphi_s|^2-|\varphi_s^{(\alpha)}|^2\big)\Big)\,\varphi_s\, \Big\|_2  \\
& +\Big\|(-\Delta)^{1/4}\Big(\frac{1}{|x|+\alpha}*\big(|\varphi_s|^2-|\varphi_s^{(\alpha)}|^2\big)\Big)(\varphi_s-\varphi_s^{(\alpha)}) \Big\|_2 \,\bigg\}\,.
\end{split}
\end{equation}

\medskip

The first term is bounded by 
\begin{equation}\label{D^1/4-splitting-1summ}
\begin{split}
\Big\|(-\Delta)^{1/4}\Big(\frac{1}{|x|}*|\varphi_s|^2\Big)(\varphi_s-\varphi_s^{(\alpha)}) \Big\|_2 \; & \lesssim \; \Big\|(-\Delta)^{1/4}\Big(\frac{1}{|x|}*|\varphi_s|^2\Big)\Big\|_6\,\|\varphi_s-\varphi_s^{(\alpha)}\|_3 \\
& \quad + \Big\|\frac{1}{|x|}*|\varphi_s|^2\Big\|_\infty\,\big\|(-\Delta)^{1/4}\big(\varphi_s-\varphi_s^{(\alpha)}\big)\,\big\|_2 
\end{split}
\end{equation}
where we used the generalized Leibniz rule (see Lemma \ref{lm:leib}). Next we observe that, by Kato's inequality (\ref{eq:kato}) and by (\ref{eq:H2-bd}), we have $\| |.|^{-1} * |\ph_s|^2 \|_{\infty} \lesssim \| \ph_s \|_{H^{1/2}}^2 \lesssim \nu^2$. This, combined with the bound 
\begin{equation}\label{D1/4-1/x-u2}
\begin{split}
\Big\|(-\Delta)^{1/4}\Big(\frac{1}{|x|}*|\varphi_s|^2\Big)\Big\|_6 \; \lesssim \; \|\varphi_s\|_3^2\,  \lesssim \| \ph_s \|_{H^{1/2}}^2 \lesssim \nu^2 
\end{split}
\end{equation}
implies (using also (\ref{distance_in_L2})) that 
\begin{equation}\label{eq:first} \Big\|(-\Delta)^{1/4}\Big(\frac{1}{|x|}*|\varphi_s|^2\Big)(\varphi_s-\varphi_s^{(\alpha)}) \Big\|_2 \lesssim \nu^2 \| \ph_s - \ph^{(\alpha)}_s \|_{H^s}\, \lesssim \nu^2 \alpha + \nu^2 \| (-\Delta)^{1/4} (\ph_s - \ph_s^{(\alpha)}) \|_2 
\end{equation}
To prove \eqref{D1/4-1/x-u2}, we rewrite $|.|^{-1}*|\varphi_s|^2=-4\pi\,(-\Delta)^{-1}|\varphi_s|^2$. Then
\begin{equation}\label{introG}
\begin{split}
\Big\|(-\Delta)^{1/4}\Big(\frac{1}{|x|}*|\varphi_s|^2\Big)\Big\|_6 \; & \lesssim \; \big\|(-\Delta)^{-3/4}|\varphi_s|^2\big\|_6\;=\;\big\|G_{3/2}*|\varphi_s|^2\big\|_6\,.
\end{split}
\end{equation}
Here $G_s$, $s\in(0,3)$, is the kernel of the operator $(-\Delta)^{-s/2}$ 
which is explicitly given by 
\begin{equation}\label{G32}
G_{3/2}(x)\;=\;c_{3/2}|x|^{-3/2}
\end{equation}
with $c_{3/2}=\pi^2\sqrt{2}/\Gamma(\frac{3}{4})$.
{F}rom (\ref{introG}), we conclude by the Littlewood-Hardy-Sobolev inequality that 
\begin{equation*}
\Big\|(-\Delta)^{1/4}\Big(\frac{1}{|x|}*|\varphi_s|^2\Big)\Big\|_6 \leq \big\|G_{3/2}*|\varphi_s|^2\big\|_6\;\lesssim\; \|\varphi_s\|_3^2
\end{equation*}
and thus \eqref{D1/4-1/x-u2} follows.

\bi

The second term on the r.h.s. of \eqref{D^1/4-splitting} is estimated again by the generalized Leibniz rule as
\begin{equation}\label{2-summ}
\begin{split}
\Big\|(-\Delta)^{1/4}&\Big(\frac{\alpha}{|x|(|x|+\alpha)}*|\varphi_s|^2\Big)(\varphi_s-\varphi_s^{(\alpha)}) \Big\|_2 \\
& \lesssim\;\Big\|(-\Delta)^{1/4}\Big(\frac{\alpha}{|x|(|x|+\alpha)}*|\varphi_s|^2\Big)\Big\|_\infty\,\|\varphi_s-\varphi_s^{(\alpha)}\|_2 \\
&\quad +\Big\|\frac{\alpha}{|x|(|x|+\alpha)}*|\varphi_s|^2\Big\|_\infty\, \big\|(-\Delta)^{1/4}\big(\varphi_s-\varphi_s^{(\alpha)}\big)\,\big\|_2\,.
\end{split}
\end{equation}
Since 
\begin{equation}\label{2b}
\begin{split}
\Big\|(-\Delta)^{1/4}\Big(\frac{\alpha}{|x|(|x|+\alpha)}*|\varphi_s|^2\Big)\Big\|_\infty \; & \leq \; \Big\|\frac{\alpha}{|x|(|x|+\alpha)}\Big\|_2\,\big\|(-\Delta)^{1/4}|\varphi_s|^2\big\|_2
\\
& \lesssim \; \alpha^{1/2} \,  \|(-\Delta)^{1/4}\varphi_s\|_3\,\|\varphi_s\|_6\;\lesssim\; \alpha^{1/2} \, \|\varphi_s\|_{H^1}^2 \lesssim \alpha^{1/2} \, \nu^2 
\end{split}
\end{equation}
and, by (\ref{eq:hardy}),
\begin{equation}\label{2c}
\begin{split}
\Big\|\frac{\alpha}{|x|(|x|+\alpha)}*|\varphi_s|^2\Big\|_\infty \; & \leq \; \alpha \,\Big\|\frac{1}{|x|^{2}}*|\varphi_s|^2\Big\|_\infty\;\lesssim\; \alpha \, \|\varphi_s\|_{H^1}^2\, \lesssim \alpha \nu^2 
\end{split}
\end{equation}
we find, using (\ref{distance_in_L2}), that 
\begin{equation}\label{D^1/4-splitting-2summ}
\Big\|(-\Delta)^{1/4}\Big(\frac{\alpha}{|x|(|x|+\alpha)}*|\varphi_s|^2\Big)(\varphi_s-\varphi_s^{(\alpha)}) \Big\|_2\;\lesssim\; \alpha^{3/2} \nu^2 + \alpha \nu^2 \big\| (-\Delta)^{1/4} \big(\varphi_s-\varphi_s^{(\alpha)}\big)\,\big\|_2\,.
\end{equation}

\bi

The third summand in \eqref{D^1/4-splitting} is estimated (again using (\ref{eq:hardy})) as
\begin{equation}\label{D^1/4-splitting-3summ}
\begin{split}
\Big\|(-\Delta)^{1/4}\Big(\frac{\alpha}{|x|(|x|+\alpha)}*|\varphi_s|^2\Big)\,\varphi_s \,\Big\|_2 \;&\lesssim\;
\Big\|\frac{\alpha}{|x|(|x|+\alpha)}*|\varphi_s|^2\Big\|_\infty\,\|(-\Delta)^{1/4}\varphi_s\|_2 \\
& \quad + \Big\|(-\Delta)^{1/	4}\frac{\alpha}{|x|(|x|+\alpha)}*|\varphi_s|^2\Big\|_3\,\|\varphi_s\|_6 \\
& \lesssim \; \alpha \, \|\varphi_s\|_{H^1}^2\,\|\varphi_s\|_{H^{1/2}} \\ & \lesssim \; \alpha \, \nu^3
\end{split}
\end{equation}
where we used  \eqref{2c}, the Sobolev inequality $\| \ph_s \|_6 \lesssim \| \ph_s \|_{H^1}$ and the bound
\begin{equation*}
\begin{split}
\Big\|(-\Delta)^{1/4}\frac{\alpha}{|x|(|x|+\alpha)}*|\varphi_s|^2\Big\|_3\; & \leq\;\alpha\,\Big\|\frac{1}{|x|^2}*\big|(-\Delta)^{1/4}|\varphi_s|^2\big|\,\Big\|_{3} \\ & \lesssim\;\alpha\,\|(-\Delta)^{1/4}(\overline{\varphi_s}\varphi_s)\|_{3/2} \\
& \lesssim \; \alpha \,\|\varphi_s\|_{H^1} \| \ph_s \|_{H^{1/2}} \,
\end{split}
\end{equation*}
where we used the Littlewood-Hardy-Sobolev inequality, and, in the last inequality, the generalized Leibniz rule (see Lemma \ref{lm:leib} below). 

\bi

The fourth summand on the r.h.s. of \eqref{D^1/4-splitting} is bounded by 
\begin{equation}\label{5summand_splitted}
\begin{split}
\Big\|(-\Delta)^{1/4}\Big(\frac{1}{|x|+\alpha}*\big(|\varphi_s|^2-&|\varphi_s^{(\alpha)}|^2 \big) \Big) \varphi_s \Big\|_2 \\
\lesssim\; & \; \Big\|\Big((-\Delta)^{1/4}\frac{1}{\,|x|+\alpha} \Big) * \big(|\varphi_s|^2-|\varphi_s^{(\alpha)}|^2 \big)  \Big\|_{2+\varepsilon} \,\|\varphi_s\|_{\frac{2(2+\varepsilon)}{\varepsilon}} \\
& + \; \Big\|\frac{1}{\,|x|+\alpha} *\big(|\varphi_s|^2-|\varphi_s^{(\alpha)}|^2 \big)  \Big\|_6 \,\|(-\Delta)^{1/4}\varphi_s\|_3 
\end{split}
\end{equation}
for arbitrary $\varepsilon>0$. The second term on the r.h.s. of (\ref{5summand_splitted}) is estimated by
\begin{equation}\label{5-3}
\begin{split}
 \Big\|\frac{1}{\,|x|+\alpha} *\big(|\varphi_s|^2-|\varphi_s^{(\alpha)}|^2 \big)  \Big\|_6 \,\|(-\Delta)^{1/4}\varphi_s\|_3\; & \lesssim \; \|\varphi_s\|_{H^1}\,\Big\|\frac{1}{|x|}*\big|\,|\varphi_s|^2-|\varphi_s^{(\alpha)}|^2 \big|\,  \Big\|_6 \\
& \lesssim \;\|\varphi_s\|_{H^1}\,\big\|\,|\varphi_s-\varphi_s^{(\alpha)}|\, \big(|\varphi_s|+|\varphi_s^{(\alpha)}| \big) \big\|_{6/5} \\
& \lesssim \;\|\varphi_s\|_{H^1}\,\big(\, \|\varphi_s\|_2+\|\varphi_s^{(\alpha)}\|_2\big)\,\|\varphi_s-\varphi_s^{(\alpha)}\|_3 \\
& \lesssim \; \nu\,\|(-\Delta)^{1/4}(\varphi_s-\varphi_s^{(\alpha)})\|_2\,.
\end{split}
\end{equation}
where we used the Sobolev inequality on the first, the Hardy-Littlewood-Sobolev inequality on the second, the H\"older inequality on the third, and, finally, again the Sobolev inequality in the fourth line. 

\medskip

As for the first term on the r.h.s. of (\ref{5summand_splitted}), we notice that 
\begin{equation}\label{5-1}
\begin{split}
\Big\|\Big((-\Delta)^{1/4}\frac{1}{\,|x|+\alpha} \Big) &* \big(|\varphi_s|^2-|\varphi_s^{(\alpha)}|^2 \big)  \Big\|_{2+\varepsilon} \,\|\varphi_s\|_{\frac{2(2+\varepsilon)}{\varepsilon}} \\
& \lesssim \; \Big\|(-\Delta)^{1/4} \frac{1}{|x|+\alpha} \,\Big\|_{2+\varepsilon}\,\big\| |\varphi_s-\varphi_s^{(\alpha)}| \big(|\varphi_s|+|\varphi_s^{(\alpha)}| \big)  \big\|_1\,\|\varphi_s\|_{H^2} \\
& \leq \; \Big\|(-\Delta)^{1/4} \frac{1}{|x|+\alpha} \,\Big\|_{2+\varepsilon}\,\|\varphi_s-\varphi_s^{(\alpha)}\|_2\,\big(\|\varphi_s\|_2+\|\varphi_s^{(\alpha)}\|_2\big)\,\,\|\varphi_s\|_{H^2} \\
&  \lesssim \; \nu \,\alpha^{1-\varepsilon}
\end{split}
\end{equation}
where in the last step we used the bound
\begin{equation}\label{bound<}
\Big\|(-\Delta)^{1/4} \frac{1}{|x|+\alpha} \,\Big\|_{2+\varepsilon}\;\lesssim\;\alpha^{-\varepsilon}\end{equation}
for all $\eps >0$. The bound \eqref{bound<} follows from the pointwise estimate
\begin{equation}\label{eq:bd-ptw} \left|\left( (-\Delta)^{1/4} \frac{1}{|x|+\alpha} \right) (x) \right| \lesssim \frac{1}{(|x|+\alpha)^{3/2}} \end{equation} valid for all $x \in \bR^3$. To show (\ref{eq:bd-ptw}), we observe that \[ \left((-\Delta) \frac{1}{(|x| + \alpha)}\right) (x) = - \frac{2\alpha}{|x| (|x| + \alpha)^3} \] and therefore
\begin{equation}\label{eq:lap4} \begin{split} \left| \left((-\Delta)^{1/4} \frac{1}{|x|+\alpha} \right) (x) \right|  &\lesssim  \left| (-\Delta)^{-3/4} \frac{\alpha}{|x| (|x|+\alpha)^3} \right| \lesssim \int \rd y \, \frac{1}{|x-y|^{3/2}} \frac{\alpha}{|y| (\alpha +|y|)^3} \, . \end{split} \end{equation}
We assume first that $|x| \geq \alpha$. {F}rom (\ref{eq:lap4}) we find  
\[ \begin{split} \left| \left((-\Delta)^{1/4} \frac{1}{|x|+\alpha} \right) (x) \right|  & \lesssim \int_{|x-y| \geq |x|/2}  \, \frac{\rd y}{|x-y|^{3/2}} \frac{\alpha}{|y| (\alpha +|y|)^3} + \int_{|x-y| \leq |x|/2} \frac{\rd y}{|x-y|^{3/2}} \frac{\alpha}{|y| (\alpha +|y|)^3} \\ & \lesssim \frac{1}{|x|^{3/2}} \int \rd y \, \frac{\alpha}{|y| (\alpha +|y|)^3} + \frac{1}{|x|^3} \int_{|x-y| \leq |x|/2} \rd y \, \frac{1}{|x-y|^{3/2}}
\end{split} \]
where we used the fact that $|x-y| \leq |x|/2$ implies $|y| \geq |x|/2$. Explicit computation implies that
\begin{equation}\label{eq:x>a} \left| \left((-\Delta)^{1/4} \frac{1}{|x|+\alpha} \right) (x) \right|  \lesssim \frac{1}{|x|^{3/2}} \lesssim \frac{1}{(|x|+\alpha)^{3/2}} \qquad \text{for all $|x| \geq \alpha$.} \end{equation} For $|x| \leq \alpha$ we notice that, by (\ref{eq:lap4}),
\begin{equation}
\label{eq:x<a} \begin{split} \left| \left((-\Delta)^{1/4} \frac{1}{|x|+\alpha} \right) (x) \right|  & \lesssim \int_{|x-y| \geq \alpha}\frac{\rd y}{|x-y|^{3/2}} \frac{\alpha}{|y| (\alpha +|y|)^3} + \int_{|x-y| \leq \alpha} \, \frac{\rd y}{|x-y|^{3/2}} \frac{\alpha}{|y| (\alpha +|y|)^3} \\ & \lesssim \frac{1}{\alpha^{3/2}}  \int \rd y \, \frac{\alpha}{|y| (\alpha +|y|)^3} + \frac{1}{\alpha^2} \int_{|x-y| \leq \alpha} \rd y \, \frac{1}{|x-y|^{3/2} |y|} \end{split}\end{equation}
Since $|x| \leq \alpha$ and $|x-y| \leq \alpha$ imply that $|y| \leq 2\alpha$, the last term is bounded, for $|x| \leq \alpha$, by 
\[ \int_{|x-y| \leq \alpha} \rd y \, \frac{1}{|x-y|^{3/2} |y|} \lesssim \int_{|x-y| \leq \alpha} \frac{\rd y}{|x-y|^{5/2}} + \int_{|y| \leq 2 \alpha} \frac{1}{|y|^{5/2}} \lesssim \alpha^{1/2} \]
Inserting back in (\ref{eq:x<a}), it follows that
\[  \left| \left((-\Delta)^{1/4} \frac{1}{|x|+\alpha} \right) (x) \right|   \lesssim \frac{1}{\alpha^{3/2}} \lesssim \frac{1}{(|x| + \alpha)^{3/2}}\qquad \text{for all $|x| \leq \alpha$}. \]
Together with (\ref{eq:x>a}), this implies (\ref{eq:bd-ptw}) and therefore (\ref{5-1}). Combining (\ref{5-3}) with (\ref{5-1}), we obtain the bound \begin{equation}\label{eq:four} \Big\|(-\Delta)^{1/4}\Big(\frac{1}{|x|+\alpha}*\big(|\varphi_s|^2-|\varphi_s^{(\alpha)}|^2 \big) \Big) \varphi_s \Big\|_2 \leq \nu \| (-\Delta)^{1/4} (\ph_s - \ph_s^{(\alpha)}) \|_2 + \nu \alpha^{1-\eps} \end{equation} for all $\eps >0$.

\medskip

The fifth summand in \eqref{D^1/4-splitting} is estimated as
\begin{equation}\label{D^1/4-splitting-4summ}
\begin{split}
\Big\|(-\Delta)^{1/4}&\Big(\frac{1}{|x|+\alpha}*\big(|\varphi_s|^2-|\varphi_s^{(\alpha)}|^2\big)\Big)(\varphi_s-\varphi_s^{(\alpha)}) \Big\|_2 \\
& \lesssim \; \Big\|\Big((-\Delta)^{1/4}\frac{1}{|x|+\alpha}\Big)*\big(|\varphi_s|^2-|\varphi_s^{(\alpha)}|^2\big) \Big\|_\infty\,\|\varphi_s-\varphi_s^{(\alpha)}\|_2 \\
& \quad + \Big\|\,\frac{1}{|x|+\alpha}*\big(|\varphi_s|^2-|\varphi_s^{(\alpha)}|^2\big) \Big\|_\infty\,\, \big\|(-\Delta)^{1/4}\big(\varphi_s-\varphi_s^{(\alpha)}\big)\,\big\|_2 \\
& \lesssim \; \Big\|(-\Delta)^{1/4}\frac{1}{|x|+\alpha}\,\Big\|_\infty\,\|\varphi_s-\varphi_s^{(\alpha)}\|_2^2 + \frac{1}{\alpha}\,\|\varphi_s-\varphi_s^{(\alpha)}\|_2\,\big\|(-\Delta)^{1/4}\big(\varphi_s-\varphi_s^{(\alpha)}\big)\,\big\|_2 \\
& \lesssim \; \alpha^{1/2} + \big\|(-\Delta)^{1/4}\big(\varphi_s-\varphi_s^{(\alpha)}\big)\,\big\|_2
\end{split}
\end{equation}
where we used the generalized Leibniz rule in the first inequality and the bounds  \eqref{distance_in_L2} and (\ref{eq:bd-ptw}) in the last inequality. 

\medskip

Inserting the estimates \eqref{eq:first}, \eqref{D^1/4-splitting-2summ}, \eqref{D^1/4-splitting-3summ}, \eqref{eq:four}, and \eqref{D^1/4-splitting-4summ} into \eqref{D^1/4-splitting} yields (using that $\nu \geq 1$ and the assumption $\alpha \leq 1$)
\begin{equation}\label{D^1/4-splitting-final}
\big\|(-\Delta)^{1/4}\big(\varphi_t-\varphi_{t}^{(\alpha)}\big)\big\|_2\;\lesssim \nu^3 \int_0^t\rd s \,\bigg\{\,\big\|(-\Delta)^{1/4}\big(\varphi_s-\varphi_s^{(\alpha)}\big)\,\big\|_2 + \alpha^{1/2}\bigg\}\,.
\end{equation}
Eq. \eqref{H1/2closeness} follows by Gronwall's lemma. 
\end{proof}

In the next corollary, we summarize the consequences of the bound on the $H^{1/2}$-norm of $\ph_t$ (and of the assumption $\ph \in H^2 (\bR^3)$ on the initial data), that  will play a crucial role in the many body analysis. 

\begin{corollary}
\label{cor:reg}
Fix $s  \geq 2$ and $\varphi\in H^s (\mathbb{R}^3)$. Let $\varphi_t$ and, for any $\alpha >0$, $\ph^{(\alpha)}_t$ be the solutions of the nonlinear Hartree equations (\ref{eq:hartree0}) and, respectively,  (\ref{eq:hara1}) with initial data $\ph_{t=0} = \ph_{t=0}^{(\alpha)} = \ph$ ($\varphi_t$ is the maximal local solution of (\ref{eq:hartree0}) in $H^{1/2} (\bR^3)$; $\ph_t^{(\alpha)}$, on the other hand, is known to exist globally in $H^{1/2} (\bR^3)$). 
Fix $T>0$ such that 
\[\kappa := \sup_{|t| \leq T} \| \ph_t \|_{H^{1/2}}  < \infty \, . \]
Then there exists $\nu = \nu (s, T,\kappa, \| \ph \|_{H^s}) < \infty$ independent of $\alpha$ such that
\[ \sup_{|t| \leq T} \| \ph^{(\alpha)}_t \|_{H^s} \leq \nu \, \]
for all $\alpha >0$ small enough. 
\end{corollary}

\begin{proof}
Since $s \geq 2$, Proposition \ref{closeness_in_H1/2} implies that
\[  \sup_{|t| \leq T} \| \ph_t^{(\alpha)} \|_{H^{1/2}} \leq 2 \kappa \]
for sufficiently small $\alpha >0$. The claim follows then by Proposition \ref{Prop:propagation_of_regularity}.
\end{proof}

To conclude this section, we state the generalized Leibniz rule for fractional derivatives. For a proof of this lemma, see \cite{GK}.

\begin{lemma}[Generalized Leibniz Rule]
\label{lm:leib} Suppose that $1<p<\infty$, $s \geq 0$, 
$\alpha \geq 0$, $\beta \geq 0$, and $1/p_i + 1/q_i = 1/p$ with $i=1,2$, $1< q_1 \leq \infty$, $1< p_i \leq \infty$. Then there exists a constant $c= c(p,p_1,p_2,s,\alpha,\beta) < \infty$ such that 
\[ \| (-\Delta)^{s/2} (fg) \|_p \leq c \left( \| (-\Delta)^{(s+\alpha)/2} f \|_{p_1} \| (-\Delta)^{-\alpha/2} g \|_{q_1} + \| (-\Delta)^{-\beta/2} f \|_{p_2} \| (-\Delta)^{(s+\beta)/2} g \|_{q_2} \right)\]
for all measurable functions $f,g$ for which the r.h.s. is finite.
\end{lemma}

\bi

\section{Fock space representation}
\label{sec:Fock}
\setcounter{equation}{0}

In this section, we introduce a Fock-space representation of our system, and we define coherent states. The bosonic Fock space over $L^2 (\bR^3, \rd x)$ is defined by
\[ \cF = \bigoplus_{n \geq 0} L^2 (\bR^3 , \rd x)^{\otimes_s n} =
\bC \oplus \bigoplus_{n \geq 1} L^2_s (\bR^{3n}, \rd x_1 \dots \rd
x_n)\, ,
\] with the convention $L^2 (\bR^3)^{\otimes_s 0} = \bC$.
Vectors in $\cF$ are sequences $\psi = \{ \psi^{(n)} \}_{n \geq 0}$
of $n$-particle wave functions $\psi^{(n)} \in L^2_s (\bR^{3n})$.
On $\cF$, we introduce the scalar product 
\[ \langle \psi_1 , \psi_2 \rangle = \sum_{n \geq 0} \langle
\psi_1^{(n)} , \psi_2^{(n)} \rangle_{L^2 (\bR^{3n})} =
\overline{\psi_1^{(0)}} \psi_2^{(0)} + \sum_{n \geq 1} \int \rd x_1
\dots \rd x_n \, \overline{\psi_1^{(n)}} (x_1 , \dots , x_n)
\psi_2^{(n)} (x_1, \dots ,x_n) \,. \] It is simple to check that, with this inner product, 
$\cF$ is a Hilbert space. States with $N$ particles and with wave function $\psi_N \in L^2_s (\bR^{3N})$ are described on $\cF$ by the sequence $\{\psi^{(n)} \}_{ n \geq 0}$ where $\psi^{(n)} = 0$ for all $n \neq N$ and $\psi^{(N)} = \psi_N$. The vector $\{1, 0, 0, \dots \} \in \cF$ is called the vacuum, and will be denoted by $\Omega$.

The number of particles operator $\cN$ acts on $\cF$ according to $(\cN
\psi)^{(n)} = n \psi^{(n)}$ for all $n \in \bN$. Eigenvectors of $\cN$ are vectors of
the form $\{ 0, \dots, 0, \psi^{(m)}, 0,  \dots \}$ with a fixed
number of particles. 

For arbitrary $f \in L^2 (\bR^3)$ we define the
creation operator $a^* (f)$ and the annihilation operator $a(f)$ on
$\cF$ by
\begin{equation}
\begin{split}
\left(a^* (f) \psi \right)^{(n)} (x_1 , \dots ,x_n) &=
\frac{1}{\sqrt n} \sum_{j=1}^n f(x_j) \psi^{(n-1)} ( x_1, \dots,
x_{j-1}, x_{j+1},
\dots , x_n) \\
\left(a (f) \psi \right)^{(n)} (x_1 , \dots ,x_n) &= \sqrt{n+1} \int
\rd x \; \overline{f (x)} \, \psi^{(n+1)} (x, x_1, \dots ,x_n) \, .
\end{split}
\end{equation}
The operators $a^* (f)$ and $a(f)$ are unbounded, densely defined,
closed operators. The creation operator $a^*(f)$ is the adjoint of
the annihilation operator $a(f)$ (note that by definition $a(f)$ is
anti-linear in $f$), and they satisfy the canonical commutation
relations \begin{equation}\label{eq:comm} [ a(f) , a^* (g) ] =
\langle f,g \rangle_{L^2 (\bR^3)}, \qquad [ a(f) , a(g)] = [ a^*
(f), a^* (g) ] = 0 \,. \end{equation} 


We will also make use of operator valued distributions $a^*_x$ and
$a_x$ ($x \in \bR^3$), defined so that \begin{equation}\begin{split}
a^* (f) &= \int \rd x \, f(x) \, a_x^* \\ a(f) & = \int \rd x \,
\overline{f (x)} \, a_x \end{split}
\end{equation}
for every $f \in L^2 (\bR^3)$. The canonical commutation relations
assume the form \[ [ a_x , a^*_y ] = \delta (x-y) \qquad [ a_x, a_y
] = [ a^*_x , a^*_y] = 0 \, .\]

The number of particle operator, expressed through the distributions
$a_x,a^*_x$, is given, formally, by
\[ \cN = \int \rd x \, a_x^* a_x \,. \]

The following standard lemma provides some useful bounds to control creation
and annihilation operators in terms of the number of particle
operator $\cN$. 
\begin{lemma}\label{lm:a-bd}
Let $f \in L^2 (\bR^3)$. Then we have 
\begin{equation}
\begin{split}
\| a(f) \psi \| &\leq \| f \| \, \| \cN^{1/2} \psi \| \qquad \text{and }  \qquad 
\| a^* (f) \psi \| \leq \| f \| \, \| \left( \cN + 1 \right)^{1/2} \, .
\psi \|
\end{split}
\end{equation}
\end{lemma}

For an arbitrary $\psi \in \cF$, we define the one-particle density
$\gamma^{(1)}_{\psi}$ associated with $\psi$ as the positive trace
class operator on $L^2 (\bR^3)$ with kernel given by
\begin{equation}\label{eq:margi} \gamma^{(1)}_{\psi} (x; y) = \frac{1}{\langle \psi,
\cN \psi \rangle} \, \langle \psi, a_y^* a_x \psi \rangle\, .
\end{equation} By definition, $\gamma_{\psi}^{(1)}$ is a positive trace
class operator on $L^2 (\bR^3)$ with $\tr \, \gamma_{\psi}^{(1)}
=1$. For every $N$-particle state with wave function $\psi_N \in
L^2_s (\bR^{3N})$ (described on $\cF$ by the sequence $\{ 0, 0,
\dots, \psi_N, 0,0, \dots \}$) it is simple to see that this
definition is equivalent to the standard definition.

\medskip

For any sequence $\alpha = (\alpha_N)$, with $\alpha_N \to 0$ as $N \to \infty$, we define the Hamiltonian $\cH^\alpha_N$ on $\cF$ by $ (\cH^\alpha_N \psi)^{(n)} =
(\cH_N^\alpha)^{(n)} \, \psi^{(n)}$, with
\[ (\cH^\alpha_N)^{(n)} = \sum_{j=1}^n (1-\Delta_{x_j})^{1/2} - \frac{\lambda}{N} \sum_{i<j}^n
\frac{1}{|x_i -x_j| + \alpha_N} \, . \] Using the distributions $a_x, a^*_x$, $\cH^\alpha_N$
can be rewritten, formally, as
\begin{equation}\label{eq:ham2} \cH^\alpha_N = \int \rd x  \, a^*_x \, (1-\Delta_x)^{1/2} \, a_x - \frac{\lambda}{2N} \int \rd x \rd y \, \frac{1}{|x-y|+\alpha_N} \,  a_x^* a_y^*
a_y a_x \, . \end{equation} By definition, the Hamiltonian $\cH^\alpha_N$
leaves sectors of $\cF$ with a fixed number of particles invariant.
Moreover, it is clear that on the $N$-particle sector, $\cH^\alpha_N$
agrees with the Hamiltonian $H^\alpha_N$. 
We will study the dynamics generated by the operator
$\cH^\alpha_N$. In particular we will consider the time evolution of
coherent states, which we introduce next.

\medskip

For $f \in L^2 (\bR^3)$, we define the Weyl-operator
\begin{equation}
W(f) = \exp \left( a^* (f) - a(f) \right) 
\end{equation}
and the coherent state $\psi (f) \in \cF$ with one-particle wave
function $f$ by $\psi (f) = W(f) \Omega$. Notice that \begin{equation}\label{eq:coh} \psi (f)= W(f) \Omega =
e^{-\| f\|^2 /2} \sum_{n \geq 0} \frac{ (a^* (f))^n}{n!} \Omega  =
e^{-\| f\|^2 /2} \sum_{n \geq 0} \frac{1}{\sqrt{n!}} \, f^{\otimes n} \,,
\end{equation}
where $f^{\otimes n}$ indicates the Fock-vector $\{ 0, \dots , 0 ,f^{\otimes n}, 0, \dots \}$. This follows from
\[ \exp (a^* (f) - a (f)) = e^{-\|f \|^2/2} \exp (a^* (f)) \exp
(-a(f)) \] which is a consequence of the fact that the commutator $[
a (f) , a^* (f)] = \| f \|^2$ commutes with $a(f)$ and $a^* (f)$.
{F}rom Eq. (\ref{eq:coh}) we see that coherent states are
superpositions of states with different number of particles (the
probability of having $n$ particles in $\psi (f)$ is given by
$e^{-\| f\|^2} \| f \|^{2n}/n!$).

\medskip

In the following standard lemma we collect some important and well known
properties of Weyl operators and coherent states.
\begin{lemma}\label{lm:coh}
Let $f,g \in L^2 (\bR^3)$.
\begin{itemize}
\item[i)] The Weyl operators satisfy the relations
\[ W(f) W(g) = W(g) W(f) e^{-2i \, \text{Im} \, \langle f,g \rangle} = W(f+g) e^{-i\, \text{Im} \, \langle f,g \rangle} \,. \]
\item[ii)] $W(f)$ is a unitary operator and
\[ W(f)^* = W(f)^{-1}  = W (-f). \]
\item[iii)] We have \[ W^* (f) a_x W(f) = a_x + f(x), \qquad \text{and} \quad W^* (f) a^*_x
W(f) = a^*_x + \overline{f} (x) \, .\]
\item[iv)] {F}rom iii) we see that coherent states are eigenvectors of annihilation operators
\[ a_x \psi (f) = f(x) \psi (f)  \qquad \Rightarrow \qquad a (g)
\psi (f) = \langle g, f \rangle_{L^2} \psi (f) \, .\]
\item[v)] The expectation of the number of particles in the coherent
state $\psi (f)$ is given by $\| f\|^2$, that is
\[ \langle \psi (f), \cN \psi (f) \rangle = \| f \|^2
\, . \] Also the variance of the number of particles in $\psi (f)$
is given by $\|f \|^2$ (the distribution of $\cN$ is Poisson), that
is
\[ \langle \psi (f), \cN^2 \psi (f) \rangle - \langle \psi (f) ,
\cN \psi (f) \rangle^2 = \| f \|^2 \, .\]
\item[vi)] Coherent states are normalized but not orthogonal to each
other. In fact
\[ \langle \psi (f) , \psi (g) \rangle = e^{-\frac{1}{2}\left( \| f
\|^2 + \| g \|^2 - 2 (f,g) \right)}  \quad \Rightarrow \quad
|\langle \psi (f) , \psi (g) \rangle| = e^{-\frac{1}{2} \| f- g
\|^2} \, .\]
\end{itemize}
\end{lemma}

\section{Time evolution of coherent states and proof of Theorem \ref{Thm:no-blowup}}
\label{sec:coh-trace}

In this section we study the time evolution of an initial coherent state $\psi (\sqrt{N} \ph) = W(\sqrt{N} \ph) \Omega$, for $\ph \in H^2 (\bR^3)$ with $\| \ph \| =1$. The expected number of particles in the coherent state $\psi (\sqrt{N} \ph)$ is $N$. Therefore, we may expect the evolution generated by $\cH_N^\alpha$ on $\psi(\sqrt{N} \ph)$ to have a mean-field character. In particular we may expect that $e^{-it \cH_N^\alpha} \psi (\sqrt{N} \ph) \simeq \psi (\sqrt{N} \ph_t)$ where $\ph_t$ solves the nonlinear Hartree equation (\ref{eq:hartree0}). We will prove that this is indeed the case, under the assumption that $\ph_t$ remains bounded in $H^{1/2} (\bR^3)$ in the time interval $[-T,T]$. 

\begin{theorem}\label{thm:cohN}
Fix $\ph \in H^2 (\bR^3)$ with $\| \ph \| =1$ and an arbitrary sequence $\alpha_N >0$ such that $\alpha_N \to 0$ as $N\to \infty$. Let $\psi (N,t) = e^{-it \cH_N^\alpha} W(\sqrt{N} \ph) \Omega$ be the evolution of the initial coherent state $W(\sqrt{N} \ph)\Omega$ generated by the Hamiltonian (\ref{eq:ham2}). Denote by $\Gamma^{(1)}_{N,t}$ the one-particle reduced density associated with $\psi (N,t)$.

\medskip

Let $\ph_t$ be the solution of the nonlinear Hartree equation (\ref{eq:hartree0}), 
with initial data $\ph_{t=0} = \ph$. Fix $T >0$ so that 
\begin{equation}\label{eq:bdpht} \kappa := \sup_{|t| \leq T} \| \ph_t \|_{H^{1/2}} < \infty \,.\end{equation} 
Then there exists $C= C (T,\kappa, \| \ph \|_{H^2}) < \infty$ such that 
\[ \tr \, \left| \Gamma^{(1)}_{N,t} - |\ph_t \rangle \langle \ph_t| \right| \leq C \left( \frac{1}{N} + \alpha_N \right) \] for all $t \in \bR$ with $|t| \leq T$.
\end{theorem}

\begin{proof} The proof of Theorem \ref{thm:cohN} is analogous to the proof of Theorem 3.1 in \cite{RS}. For completeness (and because some of these arguments will be used later on), we explain here the main steps. 

Since $|\ph_t\rangle \langle \ph_t|$ is a rank one projection, it is enough to show that
\[ \left\| \Gamma^{(1)}_{N,t} - |\ph_t \rangle \langle \ph_t| \right\|_{\text{HS}} \leq C \left(\frac{1}{N} + \alpha_N \right) \] 
where $\| A \|^2_{\text{HS}} = \tr \, A^* A$ is the Hilbert-Schmidt norm of $A$. This follows from the remark\footnote{We learned this argument from R. Seiringer} that the operator $\Gamma_{N,t}^{(1)} - |\ph_t \rangle \langle \ph_t|$ can only have one negative eigenvalue. Sine the trace vanishes, the absolute value of the negative eigenvalue must be the same as the sum of all positive eigenvalues. For this reason, the trace norm is twice the operator norm, which is of course bounded by the Hilbert-Schmid norm.

Suppose now that $\ph^{(\alpha_N)}_t$ denote the solution of the regularized Hartree equation 
\begin{equation}\label{eq:haraN} i\partial_t \ph_t^{(\alpha_N)} = \sqrt{1-\Delta} \, \ph_t^{(\alpha_N)} - \lambda \left( \frac{1}{|.|+\alpha_N} * |\ph^{(\alpha_N)}_t|^2 \right) \ph_t^{(\alpha_N)} \end{equation} with initial data $\ph_{t=0}^{(\alpha_N)} = \ph$. By Proposition \ref{closeness_in_H1/2} (see, in particular, (\ref{distance_in_L2})), and since \[ \left\|  |\ph^{(\alpha_N)}_t \rangle \langle \ph^{(\alpha_N)}_t| - |\ph_t \rangle \langle \ph_t| \right\|_{\text{HS}} \leq 2 \| \ph_t - \ph_t^{(\alpha_N)} \| \, , \] it is enough to prove that 
\begin{equation}\label{eq:HS} \| \Gamma^{(1)}_{N,t} - |\ph^{(\alpha_N)}_t \rangle \langle \ph^{(\alpha_N)}_t| \|_{\text{HS}} \leq \frac{C}{N} \end{equation} 
for a constant $C$ depending on $T,\kappa,\|\ph \|_{H^2}$.  In order to prove (\ref{eq:HS}) we write the difference of the kernels of $\Gamma^{(1)}_{N,t}$ and $|\ph^{(\alpha_N)}_t \rangle \langle \ph^{(\alpha_N)}_t|$ as (compare with (3.4) in \cite{RS})
\begin{equation}\label{eq:gamma}
\begin{split}
\Gamma^{(1)}_{N,t} (x;y) - &\ph^{(\alpha_N)}_t (x) \overline{\ph}^{(\alpha_N)}_t (y) \\ = \; &
\frac{\overline{\ph}^{(\alpha_N)}_t (y)}{\sqrt{N}} \left\langle \Omega, W^*
(\sqrt{N} \ph) e^{i\cH_N t} (a_x - \sqrt{N} \ph^{(\alpha_N)}_t(x)) e^{-i\cH_N t}
W(\sqrt{N} \ph) \Omega \right\rangle \\ &+ \frac{\ph^{(\alpha_N)}_t
(x)}{\sqrt{N}} \left\langle \Omega, W^* (\sqrt{N} \ph) e^{i\cH_N t}
(a_y^* - \sqrt{N} \, \overline{\ph}^{(\alpha_N)}_t (y)) e^{-i\cH_N t} W(\sqrt{N}
\ph) \Omega \right\rangle \\ &+ \frac{1}{N} \left\langle \Omega, W^*
(\sqrt{N} \ph) e^{i\cH_N t} ( a_y^* - \sqrt{N} \, \overline{\ph}^{(\alpha_N)}_t (y))
(a_x - \sqrt{N} \ph^{(\alpha_N)}_t (x)) e^{-i \cH_N t} W(\sqrt{N} \ph ) \Omega
\right\rangle \,.
\end{split}
\end{equation}
Then, following (3.5)-(3.8) in \cite{RS}, we can define the two-parameter group of unitary transformations $\cU_N (t;s)$ by the Schr\"odinger equation 
\begin{equation}\label{eq:cUN} i
\partial_t \cU_N (t;s) = \cL_N (t) \cU_N (t;s) \qquad \text{and}
\quad \cU_N (s;s)= 1 \end{equation} with the generator
\begin{equation}\label{eq:cLN}
\begin{split}
\cL_N (t) = & \int \rd x \, a^*_x \, (1-\Delta_x)^{1/2} a_x -\lambda \int \rd x
\, \left(\frac{1}{|.|+\alpha_N} *|\ph^{(\alpha_N)}_t|^2 \right) (x) \, a^*_x a_x \\ & -\lambda
\int \rd x \rd y \, \frac{1}{|x-y|+\alpha_N} \, \overline{\ph_t}^{(\alpha_N)} (x) \ph^{(\alpha_N)}_t (y) a^*_y a_x \\
&- \frac{\lambda}{2} \int \rd x \rd y \, \frac{1}{|x-y| +\alpha_N} \, \left( \ph^{(\alpha_N)}_t (x) \ph^{(\alpha_N)}_t (y)
a^*_x a^*_y +
\overline{\ph_t}^{(\alpha_N)} (x) \overline{\ph_t}^{(\alpha_N)} (y) a_x a_y \right) \\
&-\frac{\lambda}{\sqrt{N}} \int \rd x \rd y \, \frac{1}{|x-y| + \alpha_N} \, a_x^* \left(
\ph^{(\alpha_N)}_t (y)
a^*_y  + \overline{\ph_t}^{(\alpha_N)} (y) a_y \right) a_x \\
&-\frac{\lambda}{2N} \int \rd x \rd y \, \frac{1}{|x-y| +\alpha_N} \, a^*_x a^*_y a_y a_x \, .
\end{split}
\end{equation}
It was observed by Hepp in \cite{H} and then by Ginibre-Velo in \cite{GV} that 
\begin{equation}\label{eq:def-U} \cU^*_N (t;0) \, a_x \, \cU_N (t;0) = W^* (\sqrt{N} \ph) e^{i\cH_N^\alpha t} (a_x - \sqrt{N} \ph^{(\alpha_N)}_t (x) ) e^{-i\cH^\alpha_N t} W (\sqrt{N} \ph) \, . \end{equation}
Therefore it follows from (\ref{eq:gamma}) that 
\begin{equation}\label{eq:kerGa}
\begin{split}
\Gamma^{(1)}_{N,t} (x,y) - \ph^{(\alpha_N)}_t (x) \overline{\ph}^{(\alpha_N)}_t (y) = \;
&\frac{1}{N} \left\langle \Omega, \cU_N (t;0)^* a_y^* a_x \cU_N
(t;0) \Omega \right\rangle \\ &+ \frac{\ph^{(\alpha_N)}_t (x)}{\sqrt{N}}
\left\langle \Omega,\cU_N (t;0)^* a^*_y \cU_N (t;0) \Omega \right\rangle \\
&+ \frac{\overline{\ph}^{(\alpha_N)}_t (y)}{\sqrt{N}} \left\langle \Omega,\cU_N
(t;0)^* a_x \cU_N (t;0) \Omega \right\rangle\, .
\end{split}
\end{equation}

To get an optimal bound on the error, we also introduce, similarly to (3.9) and (3.10) in \cite{RS}, the modified evolution $\wt{\cU}_N (t;s)$ defined by the equation 
\begin{equation}\label{eq:wtcUN} i
\partial_t \wt \cU_N (t;s) = \wt\cL_N (t)\, \wt\cU_N (t;s) \qquad
\text{with} \quad \wt\cU_N (s;s) = 1 \end{equation} with the
time-dependent generator
\begin{equation}
\begin{split}
\wt \cL_N (t) = & \int \rd x \, a^*_x (1-\Delta_x)^{1/2} \, a_x -\lambda \int
\rd x \, \left( \frac{1}{|.|+\alpha_N}*|\ph^{(\alpha_N)}_t|^2 \right) (x) \, a^*_x a_x \\ &-\lambda \int \rd x
\rd y \, \frac{1}{|x-y|+\alpha_N} \, \overline{\ph}_t^{(\alpha_N)} (x) \ph^{(\alpha_N)}_t (y) a^*_y a_x \\
&- \frac{\lambda}{2} \int \rd x \rd y \, \frac{1}{|x-y|+\alpha_N} \, \left( \ph^{(\alpha_N)}_t (x) \ph^{(\alpha_N)}_t (y)
a^*_x a^*_y +
\overline{\ph}_t^{(\alpha_N)} (x) \overline{\ph}_t^{(\alpha_N)} (y) a_x a_y \right) \\
&-\frac{\lambda}{2N} \int \rd x \rd y \, \frac{1}{|x-y|+\alpha_N} \, a^*_x a^*_y a_y a_x \, .
\end{split}
\end{equation}
Since $\cU_N$ commutes with the parity operator $(-1)^{\cN}$, we have
\[ \left\langle \Omega, \wt \cU_N (t;0)^*  a_y \, \wt \cU_N (t;0) \Omega
\right\rangle = \left\langle \Omega,\wt \cU_N (t;0)^* a_x^*\,  \wt
\cU_N (t;0) \Omega \right\rangle = 0 \, .\]
Therefore, we can write 
\begin{equation*}
\begin{split}
&\Gamma^{(1)}_{N,t} (x;y) - \ph^{(\alpha_N)}_t (x) \overline{\ph}^{(\alpha_N)}_t (y) \\ & =
\frac{1}{N} \langle \Omega, \cU_N (t;0)^* a_y^* a_x \cU_N (t;0)
\Omega \rangle \\ &\;\; + \frac{\ph^{(\alpha_N)}_t (x)}{\sqrt{N}} \left(
\left\langle \Omega,\cU^*_N (t;0) a_y^* \left(\cU_N (t;0) - \wt
\cU_N (t;0)\right) \Omega \right\rangle + \left\langle \Omega,
\left(\cU^*_N (t;0) - \wt \cU^*_N (t;0) \right) a_y^* \wt \cU_N
(t;0) \Omega
\right\rangle \right)\\
&\;\; + \frac{\overline{\ph}^{(\alpha_N)}_t (y)}{\sqrt{N}} \left( \left\langle
\Omega,\cU^*_N (t;0) a_x \left( \cU_N (t;0) - \wt \cU_N (t;0)
\right) \Omega \right\rangle + \left\langle \Omega, \left(\cU^*_N
(t;0) -\wt \cU^*_N (t;0) \right) a_x \wt \cU_N (t;0) \Omega
\right\rangle \right)
\end{split}
\end{equation*}
which leads, after multiplying with a Hilbert-Schmidt observable $J$ and taking the trace, to the bound
\begin{equation}\label{eq:trJdiff}
\begin{split}
\Big| \tr \, J \left( \Gamma^{(1)}_{N,t} - |\ph^{(\alpha_N)}_t \rangle
\langle \ph^{(\alpha_N)}_t| \right) \Big|
\leq \; &\frac{\| J \|_{\text{HS}}}{N} \; \langle \cU_N (t;0) \Omega, \cN \cU_N (t;0) \Omega \rangle \\
&+ \frac{2 \| J \|_{\text{HS}}}{\sqrt{N}} \| (\cU_N (t;0) - \wt
\cU_N (t;0)) \Omega \|
\, \| (\cN+1)^{1/2} \cU_N (t;0) \Omega \| \\
&+ \frac{2 \| J \|_{\text{HS}}}{\sqrt{N}} \|(\cU_N (t;0) - \wt \cU_N
(t;0))\Omega\| \, \| (\cN+1)^{1/2} \wt \cU_N (t;0) \Omega \|\,.
\end{split}
\end{equation}
To conclude the proof of the theorem, we combine the last bound with Proposition  \ref{prop:nbd}, Proposition~\ref{prop:ntildebd}, and Proposition~\ref{prop:U-tdU} below.\end{proof}

The next proposition shows that expectations of powers of the number of particle operator, evolved with respect to the fluctuation dynamics $\cU_N$, stay bounded up to time $T$. Note that to prove Theorem \ref{thm:cohN}, it would be enough to have (\ref{eq:nbd}) for $k=1$ and for $\psi = \Omega$; for later use, however, it is useful to consider arbitrary $k \in \bN$ and $\psi \in \cF$). 

\begin{proposition}\label{prop:nbd}
Suppose that the assumptions of Theorem \ref{thm:cohN} are satisfied. Suppose moreover that the unitary evolution $\cU_N (t;s)$ is defined as in (\ref{eq:cUN}). Then, for every $k \in \bN$, there exists $C = C(k,T,\kappa,\|\ph \|_{H^2})$ such that
\begin{equation}\label{eq:nbd} \langle \cU_N (t;0) \psi , \cN^k \, \cU_N (t;0) \psi \rangle \leq C \langle \psi, (\cN + 1)^{2k+2} \psi \rangle
\end{equation} for every $\psi \in \cF$ and every $t \in \bR$ with $|t| \leq T$.
\end{proposition}

A similar estimate is also needed to control the growth of the expectation of the number of particle operator with respect to the modified dynamics $\wt{\cU}_N$ introduced in 
(\ref{eq:wtcUN}). 

\begin{proposition} \label{prop:ntildebd}
Suppose that the assumption of Theorem \ref{thm:cohN} are satisfied. Suppose moreover that the unitary evolution $\wt{\cU}_N (t;s)$ is defined as in (\ref{eq:wtcUN}). 
Then there exists $C = C(T,\kappa,\|\ph \|_{H^2})$ such that
\[ \langle \wt{\cU}_N (t;0) \Omega , \cN^3 \, \wt{\cU}_N (t;0) \Omega \rangle \leq C 
\] for every $\psi \in \cF$ and every $t \in \bR$ with $|t| \leq T$. 
\end{proposition}

Finally, we need to show that, in the second and in the third term on the r.h.s. of (\ref{eq:trJdiff}), it is possible to extract one more factor $N^{-1/2}$ from the difference between the two evolutions.

\begin{proposition} \label{prop:U-tdU}
Suppose that the assumption of Theorem \ref{thm:cohN} are satisfied. Suppose moreover that the unitary evolutions $\cU_N (t;s)$ and $\wt{\cU_N} (t;s)$ are defined as in (\ref{eq:cUN}) and in (\ref{eq:wtcUN}). Then there exists $C = C(T,\kappa,\|\ph \|_{H^2})$ such that
\[ \left\| \left( \cU_N (t;0) - \wt{\cU}_N (t;0) \right) \Omega \right\| \leq \frac{C}{\sqrt{N}}\,.\] 
\end{proposition}

The proof of these three propositions can be obtained in the exact same way as the proof of Proposition 3.3, Lemma 3.8 and Lemma 3.9 in \cite{RS}. This follows by the observation that, on the one hand, the kinetic energy (given by the second quantization of the dispersion $(1-\Delta)^{1/2}$), which is the only term in the generators $\cL_N (t)$ and $\wt{\cL}_N (t)$ which differs from the generators in \cite{RS}, commutes with the number of particle operator (and with all its powers). The other important remark is that by the assumptions in Theorem \ref{thm:cohN} (in particular, by (\ref{eq:bdpht})), and by Corollary \ref{cor:reg}, there exists $\nu= \nu (\kappa,T,\| \ph \|_{H^2}) < \infty$ independent of $\alpha_N$ such that
\[ \sup_{|t| \leq T} \| \ph_t^{(\alpha_N)} \|_{H^1} \leq \nu\,. \]
The uniform bound on the $H^1$-norm of $\ph_t$ is the only property of $\ph_t$ that is used in the proof of Proposition 3.3, Lemma 3.8 and Lemma 3.9 of \cite{RS}.

\medskip

Note that the main idea in the proof of Proposition \ref{prop:nbd} is the introduction of yet another modified dynamics $\cW_N (t;s)$ defined by 
\begin{equation}\label{eq:cWN}
i\partial_t \cW_N (t;s) = \cM_N (t) \cW_N (t;s) \qquad \text{ with } \cW_N (s;s) = 1 \qquad \text{for all $s \in \bR$,} \end{equation}
with the time-dependent generator
\begin{equation}\label{eq:cMN}
\begin{split}
\cM_N (t) = & \int \rd x \, a^*_x \, (1-\Delta_x)^{1/2} a_x -\lambda \int \rd x
\, \left(\frac{1}{|.|+\alpha_N} *|\ph^{(\alpha_N)}_t|^2 \right) (x) \, a^*_x a_x \\ & -\lambda
\int \rd x \rd y \, \frac{1}{|x-y|+\alpha_N} \, \overline{\ph_t}^{(\alpha_N)} (x) \ph^{(\alpha_N)}_t (y) a^*_y a_x \\
&- \frac{\lambda}{2} \int \rd x \rd y \, \frac{1}{|x-y| +\alpha_N} \, \left( \ph^{(\alpha_N)}_t (x) \ph^{(\alpha_N)}_t (y)
a^*_x a^*_y +
\overline{\ph_t}^{(\alpha_N)} (x) \overline{\ph_t}^{(\alpha_N)} (y) a_x a_y \right) \\
&-\frac{\lambda}{\sqrt{N}} \int \rd x \rd y \, \frac{1}{|x-y| + \alpha_N} \, a_x^* \left(
\ph^{(\alpha_N)}_t (y)
\,  {\bf 1}_{M} (\cN) \, a^*_y  + \overline{\ph_t}^{(\alpha_N)} (y) \, a_y \,  {\bf 1}_{M} (\cN)\right) a_x \\
&-\frac{\lambda}{2N} \int \rd x \rd y \, \frac{1}{|x-y| +\alpha_N} \, a^*_x a^*_y a_y a_x \, .
\end{split}
\end{equation}
where, for every $M >0$, ${\bf 1}_M (s) = 1$ for $s\leq M$, and ${\bf 1}_M (s) = 0$ if $s >M$ (${\bf 1}_M$ is the characteristic function of $(-\infty,M]$). At the end $M$ is chosen as $M=\const \cdot N$. One of the main steps in the proof of Proposition \ref{prop:nbd} is a bound for the growth of the expectation of the number of particles w.r.t. the cutoffed dynamics $\cW_N (t;s)$. We state this result explicitly, because similar ideas are used also in the next section for the proof of Theorem \ref{thm:coh-en}. The proof of the next lemma is analogous to the proof of Lemma 3.5 in \cite{RS}.

\begin{lemma}\label{lm:3.5}
Suppose that the assumptions of Proposition \ref{prop:nbd} are satisfied. Let $\cW_N$ be defined as the unitary evolution (\ref{eq:cWN}) with generator (\ref{eq:cMN}) and with $M \leq \const \cdot N$. Then, for every $k \in \bN$ there exists $C=C(\const, k,T,\kappa, \| \ph \|_{H^2})$ such that
\[ \langle \cW_N (t;0) \Omega, \cN^k \, \cW_N (t;0) \Omega \rangle \leq C \] for all $t \in \bR$ with $|t| \leq T$. 
\end{lemma}

\bigskip

Proposition \ref{prop:nbd} also allows us to conclude the proof of Theorem \ref{Thm:no-blowup}, by writing the factorized initial data as linear combinations of coherent states. We follow here the proof of Theorem 1.1 in \cite{RS} very closely; for this reason, we only discuss the main ideas.

\begin{proof}[Proof of Theorem \ref{Thm:no-blowup}]
For $\psi_N = \ph^{\otimes N} \in L^2_s (\bR^{3N})$, we write (see Lemma 4.1 in \cite{RS}) 
\begin{equation} \label{eq:repr-coh1}
\{ 0, 0, \dots, 0, \psi_N, 0, 0, \dots \} = \frac{(a^* (\ph))^N}{\sqrt{N!}} \Omega = 
d_N \int_0^{2\pi} \frac{\rd
\theta}{2\pi} \; e^{i\theta N} W ( e^{-i\theta} \sqrt{N} \ph) \Omega
\end{equation} with the constant
\begin{equation}
d_N = \frac{\sqrt{N!}}{N^{N/2} e^{-N/2}} \simeq N^{1/4}\,.
\end{equation}
The kernel of the one-particle reduced density $\gamma^{(1)}_{N,t}$ associated with $\psi_{N,t} = e^{-iH_N^\alpha t} \psi_N$ is therefore given by
\begin{equation}
\gamma^{(1)}_{N,t} (x;y) =  \frac{d^2_N}{N} \int_0^{2\pi} \frac{\rd \theta_1}{2\pi} \int_0^{2\pi}
\frac{\rd \theta_2}{2\pi} \, e^{-i\theta_1 N} e^{i\theta_2 N}
\langle W(e^{-i\theta_1} \sqrt{N} \ph) \Omega, a_y^* (t) a_x (t)
W(e^{-i\theta_2} \sqrt{N} \ph) \Omega \rangle
\end{equation}
where we introduced the notation $a_x (t) = e^{i\cH_N t} a_x e^{-i\cH_N t}$. As in (4.5)-(4.7) of \cite{RS}, we find 
\begin{equation}\label{eq:repr-coh2}
\begin{split}
\gamma^{(1)}_{N,t} (x;y) - \overline{\ph}^{(\alpha_N)}_t (y) \ph^{(\alpha_N)}_t (x) = \; & \frac{d^2_N}{N} \int_0^{2\pi}
\frac{\rd \theta_1}{2\pi} \int_0^{2\pi} \frac{\rd \theta_2}{2\pi} \,
e^{-i\theta_1 N} e^{i\theta_2 N}  \left\langle \cU_N^{\theta_1} (t;0) \Omega, \, a_y^* a_x \, \cU_N^{\theta_2} (t;0) \Omega \right\rangle \\ &
+ \frac{d_N \ph^{(\alpha_N)}_t (x)}{\sqrt{N}} \int_0^{2\pi} \frac{\rd
\theta}{2\pi} \left \langle\cU_N^{\theta} (t;0) \Omega, a_y^* \, \ph^{\otimes (N-1)} \right\rangle \\
&+\frac{d_N \overline{\ph}^{(\alpha_N)}_t (y)}{\sqrt{N}} \int_0^{2\pi} \frac{\rd
\theta}{2\pi} \, \left\langle \cU_N^{\theta} (t;0) \Omega ,  a_x \, \ph^{\otimes (N-1)} \right\rangle
\end{split}
\end{equation}
where $\ph^{\otimes (N-1)}$ actually denotes the vector $\{ 0, \dots , 0,  \ph^{\otimes (N-1)}, 0 ,\dots \} \in \cF$ and where $\cU_N^{\theta} (t;0)$ is defined as the unitary evolution in (\ref{eq:cUN}), with $\ph^{(\alpha_N)}_t$ replaced by $e^{i\theta} \ph^{(\alpha_N)}_t$ (it is important to observe that if $\ph^{(\alpha_N)}_t$ solves the nonlinear Hartree equation, also $e^{i\theta} \ph^{(\alpha_N)}_t$ is a solution). Therefore, we conclude that 
\begin{equation}\label{eq:thmconc}
\begin{split}
\int \rd x \rd y \;& |\gamma^{(1)}_{N,t} (x;y) - \ph^{(\alpha_N)}_t (x)
\overline{\ph}^{(\alpha_N)}_t (y) \Big|^2  \\ \leq \; & 2\frac{d^4_N}{N^2} \,
\int_0^{2\pi} \frac{\rd \theta_1}{2\pi} \int_0^{2\pi} \frac{\rd
\theta_2}{2\pi} \; \| \cN^{1/2} \cU^{\theta_1}_{N} (t;0) \Omega \|^2
\, \| \cN^{1/2} \cU^{\theta_2}_{N} (t;0) \Omega \|^2 + \frac{4}{N} \int \rd x |f_N (x)|^2
\end{split}
\end{equation}
with 
\begin{equation} f_N (x) = d^2_N \int_0^{2\pi} \frac{\rd
\theta}{2\pi} \, \left\langle \cU_N^{\theta} (t;0) \Omega ,  a_x \, \ph^{\otimes (N-1)} \right\rangle \, . 
\end{equation}
Proceeding exactly as in Lemma 4.2 of \cite{RS}, we find a constant $C=C(T,\kappa,\|\ph \|_{H^2})$ such that
\[ \int dx \, |f_N (x)|^2  \leq C \] 
uniformly in $N$. Proposition \ref{prop:nbd} implies therefore that there exists $C=C(T,\kappa,\|\ph \|_{H^2})$ with 
\[ \left\| \gamma^{(1)}_{N,t} - |\ph_t^{(\alpha)} \rangle \langle \ph_t^{(\alpha)}| \right\|_{\text{HS}} \leq \frac{C}{\sqrt{N}} \] 
Therefore, it follows from Proposition \ref{closeness_in_H1/2} (in particular, (\ref{distance_in_L2})) that 
\[ \left\| \gamma^{(1)}_{N,t} - |\ph_t \rangle \langle \ph_t| \right\|_{\text{HS}} \leq C \left( \frac{1}{\sqrt{N}} +\alpha_N \right) \, . \] 
Since $|\ph_t \rangle \langle \ph_t|$ is a rank one projection, and since $\tr \, \gamma^{(1)}_{N,t} = \tr \, |\ph_t \rangle \langle \ph_t| = 1$, the trace norm of the difference $\gamma^{(1)}_{N,t} - |\ph_t \rangle \langle \ph_t|$ is at most two times its Hilbert-Schmidt norm. This completes the proof of the theorem.
\end{proof}

\section{Convergence in energy}
\label{sec:conv-en}

We study again the evolution of initial coherent states in the Fock space. This time, we establish the convergence of the one-particle reduced density towards the solution of the Hartree equation in the energy norm. As a consequence, we obtain a proof of Theorem \ref{thm:conv-en}. 

\medskip

\begin{theorem}\label{thm:coh-en}
Fix $\ph \in H^2 (\bR^3)$ with $\| \ph \| =1$ and a sequence $\alpha_N >0$ such that $\alpha_N \to 0$ and $N^\beta \alpha_N \to \infty$ as $N \to \infty$, for an appropriate $\beta >0$. Let $\psi (N,t) = e^{-it \cH_N^\alpha} W(\sqrt{N} \ph) \Omega$ be the evolution of the initial coherent state $W(\sqrt{N} \ph)\Omega$ generated by the Hamiltonian (\ref{eq:ham2}). Denote by $\Gamma^{(1)}_{N,t}$ the one-particle reduced density associated with $\psi (N,t)$.

\medskip

Let $\ph_t$ be the solution of the nonlinear Hartree equation (\ref{eq:hartree0}) 
with initial data $\ph_{t=0} = \ph$. Fix $T >0$ so that 
\begin{equation}\label{eq:bdpht-T} \kappa := \sup_{|t| \leq T} \| \ph_t \|_{H^{1/2}} < \infty \,.\end{equation} 
Then there exists $C= C (\beta, T,\kappa, \| \ph \|_{H^2}) < \infty$ such that 
\begin{equation}\label{eq:coh-en}  \left\| (1-\Delta)^{1/4} \left(\Gamma^{(1)}_{N,t} - |\ph_t \rangle \langle \ph_t| \right) (1-\Delta)^{1/4} \right\|_{\text{HS}}   \leq C \left( \frac{1}{\sqrt{N}} + \alpha^{1/2}_N \right) \end{equation} for all $t \in \bR$ with $|t| \leq T$.
\end{theorem}

{\it Remark.} {F}rom (\ref{eq:coh-en}) we can conclude, using arguments similar to the ones used below in the proof of Theorem \ref{thm:conv-en} (starting from Eq. (\ref{eq:HS-coh-en})), that 
\[ \tr \, \left| (1-\Delta)^{1/4} \left(\Gamma^{(1)}_{N,t} - |\ph_t \rangle \langle \ph_t| \right) (1-\Delta)^{1/4} \right| \leq C \left( \frac{1}{\sqrt{N}} + \alpha_N^{1/2} \right) \, . \]

\begin{proof}
Denote now by $\ph^{(\alpha_N)}_t$ the solution of the regularized Hartree equation
(\ref{eq:haraN}) with initial data $\ph^{(\alpha_N)}_{t=0} = \ph$. Since 
\[ \left\| (1-\Delta)^{1/4} \left( |\ph_t \rangle \langle \ph_t| - |\ph^{(\alpha_N)}_t \rangle \langle \ph^{(\alpha_N)}_t| \right) (1-\Delta)^{1/4} \right\|_{\text{HS}} \lesssim \| \ph_t - \ph_t^{(\alpha_N)} \|_{H^{1/2}} \] and using Proposition \ref{closeness_in_H1/2} (see, in particular, (\ref{H1/2closeness})), it is enough to prove that there exists a constant $C=C(\kappa,T,\| \ph\|_{H^2})$ such that
\begin{equation}\label{eq:DGaN} \left\| (1-\Delta)^{1/4} \left(\Gamma^{(1)}_{N,t} - |\ph^{(\alpha_N)}_t \rangle \langle \ph^{(\alpha_N)}_t| \right) (1-\Delta)^{1/4} \right\|_{\text{HS}} \leq \frac{C}{\sqrt{N}} . \end{equation} 
To show (\ref{eq:DGaN}), we use again the representation (\ref{eq:kerGa}) for the kernel of $\Gamma^{(1)}_{N,t}$, which implies that
\begin{equation}\label{eq:kerDGa}
\begin{split}
\Big( (1-\Delta)^{1/4} &\left(\Gamma^{(1)}_{N,t} - |\ph^{(\alpha_N)}_t\rangle \langle \ph^{(\alpha_N)}_t| \right) (1-\Delta)^{1/4} \Big) (x,y) \\ = \;
&\frac{1}{N} \left\langle \Omega, \cU_N (t;0)^* (1-\Delta_y)^{1/4} a_y^*  \, (1-\Delta_x)^{1/4} \, a_x \, \cU_N
(t;0) \Omega \right\rangle \\ &+ \frac{(1-\Delta)^{1/4}\,  \ph^{(\alpha_N)}_t (x)}{\sqrt{N}}
\left\langle \Omega,\cU_N (t;0)^* (1-\Delta_y)^{1/4} a^*_y \, \cU_N (t;0) \Omega \right\rangle \\
&+ \frac{(1-\Delta)^{1/4} \,  \overline{\ph}^{(\alpha_N)}_t (y)}{\sqrt{N}} \left\langle \Omega,\cU_N
(t;0)^* (1-\Delta_x)^{1/4} a_x \, \cU_N (t;0) \Omega \right\rangle\,
\end{split}
\end{equation}
With a Schwarz inequality, we find 
\begin{equation}
\begin{split}
\big\| (1-\Delta)^{1/4}&(\Gamma_{N,t}^{(1)}-|\varphi_t\ra\la\varphi_t| )(1-\Delta)^{1/4}\big\|_{\mathrm{HS}}^2 \\
& \lesssim \;\frac{1}{\;N^2}\,\big\la\UU\Omega,\cK\,\UU\Omega\big\ra^2+\frac{\|\varphi^{(\alpha_N)}_t\|_{H^{1/2}}^2}{N}\,\big\la\UU\Omega,\cK\,\UU\Omega\big\ra\, ,
\end{split}
\end{equation} 
where we defined 
\begin{equation}
\cK\;=\;\int\rd x \,(1-\Delta_x)^{1/4}a_x^*(1-\Delta_x)^{1/4} a_x
\end{equation} 
to be the kinetic energy operator. The theorem follows now from Proposition \ref{prop:Kbd} below.
\end{proof}

The key point, in the proof of Theorem \ref{thm:coh-en}, and also in the proof of Theorem \ref{thm:conv-en}, is the following proposition, which controls the growth of the expectation of the kinetic energy with respect to the fluctuation dynamics $\cU_N$.

\begin{proposition}\label{prop:Kbd}
Suppose that the assumptions of Theorem \ref{thm:coh-en} are satisfied. Suppose moreover that the unitary evolution $\cU_N (t;s)$ is defined as in (\ref{eq:cUN}). Then there exists $C = C(T,\kappa,\|\ph \|_{H^2})$ such that
\begin{equation}\label{eq:Kbd} \langle \cU_N (t;0) \Omega , \cK \, \cU_N (t;0) \Omega \rangle \leq C 
\end{equation} for $t \in \bR$ with $|t| \leq T$. 
\end{proposition}

The proof of Proposition \ref{prop:Kbd} is given in Section \ref{sec:Kbd}.
The bound on the growth of the expectation of $\cK$ can also be used to conclude the proof of Theorem \ref{thm:conv-en}.

\begin{proof}[Proof of Theorem \ref{thm:conv-en}] 
Using the representation (\ref{eq:repr-coh2}), we obtain 
\begin{equation}\label{intabs}
\begin{split}
\Big|\Big( (1-\Delta)^{1/4} \big( &\gamma^{(1)}_{N,t} - |\ph^{(\alpha_N)}_t \rangle \langle \ph^{(\alpha_N)}_t| \big) (1-\Delta)^{1/4} \Big) (x,y) \Big| \\
&\lesssim \frac{\;d^2_N}{\,N} \int_0^{2\pi} \!\!\int_0^{2\pi} \!\rd \theta_1\rd \theta_2 \,\big\|(1-\Delta_x)^{1/4} a_x \, \cU_N^{\theta_1} (t;0) \Omega\big\| \,\big\|(1-\Delta_y)^{1/4} a_y \,\cU^{\theta_2}_N (t;0) \Omega \big\| \\
& \quad + \frac{\:d_N}{\sqrt{N}}\,\big|(1-\Delta_x)^{1/4}\varphi^{(\alpha_N)}_t(x)\big| \int_0^{2\pi}\!\rd\theta_1\,\big\|(1-\Delta_y)^{1/4} a_y \, \cU_N^{\theta_1} (t;0) \Omega\big\| 
\end{split}
\end{equation}
where $\cU_N^{\theta} (t;0)$ is defined as the unitary evolution in (\ref{eq:cUN}), with $\ph^{(\alpha_N)}_t$ replaced by $e^{i\theta} \ph^{(\alpha_N)}_t$. 
Taking the square and integrating over $x,y$, we find
\begin{equation}\label{intsquared}
\begin{split}
\int\rd x\,\rd y &\,\Big|\Big( (1-\Delta)^{1/4} \big( \gamma^{(1)}_{N,t} - |\ph^{(\alpha_N)}_t \rangle \langle \ph^{(\alpha_N)}_t| \big) (1-\Delta)^{1/4} \Big) (x,y) \Big|^2 \\
&\lesssim\;\frac{1}{N} \bigg(\int_{0}^{2\pi}\!\!\rd\theta\,\big\la \cU^{\theta}_N (t;0) \Omega,\cK\,\cU^{\theta}_N (t;0) \Omega\big\ra\bigg)^{\!2} + \frac{1}{\sqrt{N}} \int_{0}^{2\pi}\!\!\rd\theta\,\big\la \cU^{\theta}_N (t;0) \Omega,\cK\, \cU^{\theta}_N (t;0) \Omega\big\ra \,.
\end{split}
\end{equation} 
Proposition \ref{prop:Kbd} implies that there exists $C=C(T,\kappa, \| \ph \|_{H^2})$ such that
\begin{equation}\label{eq:HS-coh-en}
\Big\| (1-\Delta)^{1/4} \big( \gamma^{(1)}_{N,t} - |\ph^{(\alpha_N)}_t \rangle \langle \ph^{(\alpha_N)}_t| \big) (1-\Delta)^{1/4} \Big\|_{\text{HS}} \leq \frac{C}{N^{1/4}} \, . 
\end{equation}
By Proposition \ref{closeness_in_H1/2} we obtain therefore
\begin{equation}\label{eq:con-fact}
\Big\| (1-\Delta)^{1/4} \big( \gamma^{(1)}_{N,t} - |\ph_t \rangle \langle \ph_t| \big) (1-\Delta)^{1/4} \Big\|_{\text{HS}} \leq C \left(\frac{1}{N^{1/4}} + \alpha_N^{1/2} \right) \, . 
\end{equation}

Note also that (\ref{intabs}) implies that 
\begin{equation}
\begin{split}
\int \rd x \, \Big|\Big( (1-\Delta)^{1/4} \big( &\gamma^{(1)}_{N,t} - |\ph^{(\alpha_N)}_t \rangle \langle \ph^{(\alpha_N)}_t| \big) (1-\Delta)^{1/4} \Big) (x,x) \Big| \\
&\lesssim \frac{1}{N^{1/4}} \bigg(\| \ph^{(\alpha_N)}_t \|_{H^{1/2}}  + \int_{0}^{2\pi}\!\!\rd\theta\,\big\la \cU^{\theta}_N (t;0) \Omega,\cK\,\cU^{\theta}_N (t;0) \Omega\big\ra\bigg)
\end{split}
\end{equation}
and therefore, by Proposition \ref{prop:Kbd},
\begin{equation}
 \left| \tr \, (1-\Delta)^{1/4} \gamma^{(1)}_{N,t} (1-\Delta)^{1/4} - \|(1-\Delta)^{1/4} \ph^{(\alpha_N)}_t \|^2 \right| \lesssim \frac{C}{N^{1/4}} \,. \end{equation}
Again, Proposition \ref{closeness_in_H1/2} implies that 
\begin{equation}\label{eq:con-tr} \left| \tr \, (1-\Delta)^{1/4} \gamma^{(1)}_{N,t} (1-\Delta)^{1/4} - \|(1-\Delta)^{1/4} \ph_t \|^2 \right| \lesssim C \left( \frac{1}{N^{1/4}} + \alpha_N^{1/2} \right) \,. \end{equation}
Last equation, together with (\ref{eq:con-fact}), implies that
\begin{equation}\label{convergenceHSnormalised}
\begin{split}
& \Bigg\|\frac{(1-\Delta)^{1/4}\gamma^{(1)}_{N,t}(1-\Delta)^{1/4}}{\;\tr[(1-\Delta)^{1/4}\gamma^{(1)}_{N,t}(1-\Delta)^{1/4}]\;}-\frac{\;|(1-\Delta)^{1/4}\varphi_t\ra\la(1-\Delta)^{1/4}\varphi_t|\;}{\|(1-\Delta)^{1/4}\varphi_t\|_2^2}\Bigg\|_{\mathrm{HS}} \\
& \qquad \leq\;\;\frac{\;\big\| (1-\Delta)^{1/4}\big(\gamma^{(1)}_{N,t} - |\ph_t \rangle \langle \ph_t|\big)(1-\Delta)^{1/4}\big\|_{\mathrm{HS}}}{\tr[(1-\Delta)^{1/4}\gamma^{(1)}_{N,t}(1-\Delta)^{1/4}]}\;+\\
&\qquad \qquad +\;\frac{\;\big|\tr[(1-\Delta)^{1/4}\gamma^{(1)}_{N,t}(1-\Delta)^{1/4}]-\|(1-\Delta)^{1/4}\varphi_t\|_2^2 \big|\;}{\;\tr[(1-\Delta)^{1/4}\gamma^{(1)}_{N,t}(1-\Delta)^{1/4}]} \\
& \qquad \leq \; C \left( \frac{1}{N^{1/4}} + \alpha_N^{1/2} \right).
\end{split}
\end{equation} 
On the l.h.s.~of \eqref{convergenceHSnormalised} we are now comparing a density matrix with a rank-one projection. The trace norm of their difference is at most twice the corresponding Hilbert-Schmidt norm and thus we find
\begin{equation}
\tr \; \left| \frac{(1-\Delta)^{1/4} \gamma^{(1)}_{N,t}(1-\Delta)^{1/4}}{\;\tr[(1-\Delta)^{1/4}\gamma^{(1)}_{N,t}(1-\Delta)^{1/4}]\;}-\frac{\;(1-\Delta)^{1/4} |\varphi_t\ra\la \varphi_t|(1-\Delta)^{1/4}\;}{\|(1-\Delta)^{1/4}\varphi_t\|_2^2} \right| \leq
C \left( \frac{1}{N^{1/4}} + \alpha_N^{1/2} \right)\,.
\end{equation}
Combining last equation with (\ref{eq:con-tr}), we finally obtain
\begin{equation*}
\begin{split}
\tr \; \Big| (1-\Delta)^{1/4} &\big(\gamma^{(1)}_{N,t} - |\ph_t \rangle \langle \ph_t|\big)(1-\Delta)^{1/4}\Big|
\; \\
\lesssim \; &\| (1-\Delta)^{1/4} \ph_t \|^2 \, \tr \, \left| \frac{(1-\Delta)^{1/4} \gamma^{(1)}_{N,t}(1-\Delta)^{1/4}}{ \| (1-\Delta)^{1/4} \ph_t \|^2}  - \frac{(1-\Delta)^{1/4} |\ph_t \rangle \langle \ph_t| (1-\Delta)^{1/4}}{\| (1-\Delta)^{1/4} \ph_t \|^2} \right| 
\\ 
\lesssim \; &\| (1-\Delta)^{1/4} \ph_t \|^2 \, \tr \, \left| \frac{(1-\Delta)^{1/4} \gamma^{(1)}_{N,t}(1-\Delta)^{1/4}}{\tr \, (1-\Delta)^{1/4} \gamma^{(1)}_{N,t} (1-\Delta)^{1/4}}  - \frac{(1-\Delta)^{1/4} |\ph_t \rangle \langle \ph_t| (1-\Delta)^{1/4}}{\| (1-\Delta)^{1/4} \ph_t \|^2} \right| \\ &+ \left| \tr \, (1-\Delta)^{1/4} \gamma^{(1)}_{N,t} (1-\Delta)^{1/4} - \|(1-\Delta)^{1/4} \ph_t \|^2 \right| \\ \leq \; & C \left( \frac{1}{N^{1/4}} + \alpha_N^{1/2} \right)
\end{split}\end{equation*}
which concludes the proof of the theorem. 
\end{proof}

\section{Control of the growth of the kinetic energy}
\label{sec:Kbd}

The goal of this section is to prove Proposition \ref{prop:Kbd}, which gives control on the growth of the expectation of the kinetic energy operator $\cK$ with respect to the fluctuation dynamics $\cU_N (t;0)$. 

\bi

\begin{proof}[Proof of Proposition \ref{prop:Kbd}]
Recall the definition (\ref{eq:cUN}) of the unitary maps $\cU_N (t;s)$ describing the evolution of the fluctuations; note that the generator $\cL_N (t)$ of $\cU_N (t;s)$ is defined in terms of the solution $\ph_t^{(\alpha_N)}$ of the regularized Hartree equation \begin{equation}\label{eq:reg-aN} i\partial_t \ph_t^{(\alpha_N)} = \sqrt{1-\Delta} \, \ph_t^{(\alpha_N)} - \lambda \left( \frac{1}{|.|+\alpha_N} * |\ph^{(\alpha_N)}_t|^2 \right) \ph_t^{(\alpha_N)} \, .\end{equation}
In the rest of this section, we will use the shorthand notation $\phi_t \equiv \ph_t^{(\alpha_N)}$. By (\ref{eq:bdpht-T}), by the assumption $\ph \in H^2 (\bR^3)$ on the initial data, and by Corollary \ref{cor:reg}, there exists $\nu = \nu (T,\kappa,\| \ph \|_{H^2})$ such that 
\begin{equation}\label{eq:bdpha} \sup_{|t| \leq T} \, \| \phi_t \|_{H^2} \leq \nu\, \end{equation}
uniformly in $N$ ($\phi_t = \ph_t^{(\alpha_N)}$ depends on $N$ through the cutoff $\alpha_N$). 

\medskip

%

We compare the growth of $\cK$ along the fluctuation dynamics $\mathcal{U}_N$ and along a new dynamics $\widetilde{\mathcal{W}}_N$ defined through the equation 
\begin{equation}\label{eq:cUNtilde} 
i\partial_t\,\wt{\cW}_N (t;s) \;=\;\tLL\,\wt{\cW}_N (t;s) \qquad\mathrm{with}\quad\widetilde{\mathcal{W}}_N(s;s)= 1 \qquad \text{for all $s \in \bR$,}
\end{equation} 
with the time-dependent generator
\begin{equation}\label{def:Ltilde}
\begin{split}
\tLL\;:=&\;\int \rd x \, a_x^* (1-\Delta_x)^{1/2} a_x \;-\; \lambda \int\rd x\,\big(\frac{1}{|.|+\alpha_N} * |\phi_t|^2\big) \,a_x^* a_x \\
&-\lambda \int\rd x \,\rd y\,\frac{1}{|x-y|+\alpha_N}\,\overline{\phi}_t (x)\,\phi_t(y)\,a_y^* a_x \\
&-\frac{\lambda}{2}\int\rd x \,\rd y\,\frac{1}{|x-y|+\alpha_N} \, \big\{\phi_t(x)\phi_t(y)\,a_x^* a_y^*+\overline{\phi}_t(x)\,\overline{\phi}_t(y) \,a_x a_y\big\} \\
&-\frac{\lambda}{\,\sqrt{N}\,}\int\rd x \,\rd y\,\frac{1}{|x-y| + \alpha_N} \, \big\{\phi_t(y) \,a_x^* a_y^* \,{\bf 1}_{\vartheta N}(\cN)\,a_x+\overline{\phi}_t(y) \,a_x^* \, {\bf 1}_{\vartheta N} (\cN)\,a_x a_y\big\} 
\\ &-\frac{\lambda}{2N}\int\rd x \,\rd y\,\frac{1}{|x-y|+\alpha_N} \,a_x^* a_y^* \, {\bf 1}_{\vartheta N} (\cN) \,a_y a_x
\end{split}
\end{equation} 
where ${\bf 1}_{\vartheta N} (s)$ is the characteristic function of the interval $(-\infty,\vartheta N]$, that is ${\bf 1}_{\vartheta N} (s) = 1$ if $s \leq \vartheta N$ and ${\bf 1}_{\vartheta N} (s) = 0$ otherwise. Here $\vartheta \leq 1$ will be fixed later to be sufficiently small. 

\bi

We split
\begin{equation}\label{UKUsplit}
\begin{split}
\big\la\UU\Omega,\cK\,\UU\Omega\big\ra  = \, &\big\la\tUU\Omega,\cK\,\tUU\Omega\big\ra  + \big\la(\UU-\tUU)\Omega,\cK\,\tUU\Omega\big\ra \\
&+ \big\la\UU\Omega,\cK (\UU-\tUU)\Omega\big\ra  \\
\leq \, &  \big\la\tUU\Omega,\cK\,\tUU\Omega\big\ra + \big\|\cK\,\tUU\Omega\big\|\,\big\|(\UU-\tUU)\Omega\big\|\, \\
& + \big\|\cK\,\UU\Omega\big\|\,\big\|(\UU-\tUU)\Omega\big\|  \, . 
\end{split}
\end{equation} 
Proposition \ref{prop:Kbd} now follows from Proposition \ref{prop:K2bdt}, Proposition \ref{prop:K2bd}, and Proposition \ref{prop:U-Utilde} and from the assumption that $N^\beta \alpha_N \to \infty$ for some $\beta >0$.
\end{proof}

The first ingredient in the proof of Proposition \ref{prop:Kbd} is a bound for the growth of the kinetic energy $\cK$ and of its square w.r.t. the cutoffed evolution $\wt{\cW}_N (t;s)$. This is the content of the next Proposition, which will be proven in Section~\ref{sec:K2bdt}.

\begin{proposition}\label{prop:K2bdt}
Suppose that the assumptions of Proposition \ref{prop:Kbd} are satisfied (but here the assumption $N^\beta \alpha_N \to \infty$ for some $\beta >0$ will not be used). Let the evolution $\wt{\cW}_N (t;s)$ be defined according to (\ref{eq:cUNtilde}), with generator (\ref{def:Ltilde}), and suppose that $\vartheta>0$ is small enough. 
Then there exists $C = C(\vartheta, T,\kappa,\|\ph \|_{H^2})$ such that
\begin{equation}\label{eq:Kbd-tilde} 
\langle \wt{\cW}_N (t;0) \Omega , \cK^2 \, \wt{\cW}_N (t;0) \Omega \rangle \leq C
\end{equation} for all $t \in \bR$ with $|t| \leq T$. 
\end{proposition}

The second ingredient to prove Proposition \ref{prop:Kbd} is a weak bound on the growth of the expectation of $\cK^2$ with respect to the dynamics $\cU_N (t;s)$; the proof of the following proposition is given in Section~\ref{sec:K2bd}.

\begin{proposition}\label{prop:K2bd}
Suppose that the assumptions of Proposition \ref{prop:Kbd} are satisfied. Then there exists $C = C(T,\kappa,\|\ph \|_{H^2})$ such that
\begin{equation}\label{eq:K2-bd-w} 
\langle \cU_N (t;0) \Omega , \cK^2 \, \cU_N (t;0) \Omega \rangle \leq C \left( N^2 + \frac{N^2}{\alpha_N^2} \right) 
\end{equation} for all $t \in \bR$ with $|t| \leq T$. 
\end{proposition}

Finally, we need to compare the two dynamics $\cU_N (t;s)$ and $\wt{\cW}_N (t;s)$. The next proposition is shown in Section \ref{sec:U-Utilde}.

\begin{proposition}\label{prop:U-Utilde}
Suppose that the assumptions of Proposition \ref{prop:Kbd} are satisfied. Let the evolution $\wt{\cW}_N (t;s)$ be defined according to (\ref{eq:cUNtilde}), with generator (\ref{def:Ltilde}), and with $\vartheta>0$. Then, for any $k \in \bN$, there exists $C = C(k, \vartheta, T,\kappa,\|\ph \|_{H^2})$ such that
\begin{equation}\label{eq:bd-diff} 
\left\| \left( \cU_N (t;0) - \wt{\cW}_N (t;0) \right) \Omega \right\| \leq \frac{C}{N^k} \, \left(1+\frac{1}{\alpha_N} \right) 
\end{equation} for all $t \in \bR$ with $|t| \leq T$. 
\end{proposition}

\subsection{Growth of $\cK^2$ with respect to regularized dynamics}
\label{sec:K2bdt}

The goal of this section is to prove Proposition \ref{prop:K2bdt}. We will make systematic use of the bound (\ref{eq:bdpha}) (recall that in this section we use the shorthand notation $\phi_t \equiv \ph^{(\alpha)}_t$).

\bigskip

Observe that, by the definition \eqref{def:Ltilde} of $\wt{\cM}_N (t)$, we have
\begin{equation}
\begin{split}
\cK^2\;\lesssim\;  &\widetilde{\cM}_N^2 (t) + \left( \int\rd x\,\left(\frac{1}{|.| +\alpha_N} * |\phi_t|^2\right) \,a_x^* a_x\right)^2 \\
& \; + \Big( \int\rd x \,\rd y\, \frac{1}{|x-y|+\alpha_N} \,\overline{\phi}_t(x)\,\phi_t(y)\,a_y^* a_x \Big)^2 \\
& \; + \Big( \frac{1}{2}\int\rd x \,\rd y\,\frac{1}{|x-y|+\alpha_N} \big\{\phi_t(x)\phi_t(y)\,a_x^* a_y^*+\mathrm{h.c.}\big\} \Big)^2 \\
& \; + \Big( \frac{1}{\,\sqrt{N}\,}\int\rd x \,\rd y\, \frac{1}{|x-y|+\alpha_N} \big\{\phi_t(y) \,a_x^* a_y^* \,{\bf 1}_{\vartheta N} (\cN)\,a_x+\mathrm{h.c.}\big\} \Big)^2\, \\ 
& \; +\Big( \frac{1}{2N} \int\rd x \,\rd y\, \frac{1}{|x-y|+\alpha_N}  \, a_x^* a_y^* \, {\bf 1}_{\vartheta N} (\cN) a_y a_x  \Big)^2
\end{split}
\end{equation} 
where $\text{h.c.}$ denotes the hermitian conjugate. First of all, we note that, by Lemma \ref{Lemma:VtK} below, the last term is bounded by
\[ \Big( \frac{1}{2N} \int\rd x \,\rd y\, \frac{1}{|x-y|+\alpha_N}  \, a_x^* a_y^* \, {\bf 1}_{\vartheta N} (\cN) a_y a_x  \Big)^2 \lesssim \, \vartheta^2 \,  \cK^2 \,  \]
Therefore, choosing $\vartheta >0$ sufficiently small, we find
\begin{equation}\label{K^2=Ltilde^2+split_0}
\begin{split}
\cK^2\;\lesssim \;  &\widetilde{\cM}_N^2 (t) + \left( \int\rd x\,\left(\frac{1}{|.| +\alpha_N} * |\phi_t|^2\right) \,a_x^* a_x\right)^2 \\
& \; + \Big( \int\rd x \,\rd y\, \frac{1}{|x-y|+\alpha_N} \,\overline{\phi}_t(x)\,\phi_t(y)\,a_y^* a_x \Big)^2 \\
& \; + \Big( \frac{1}{2}\int\rd x \,\rd y\,\frac{1}{|x-y|+\alpha_N} \big\{\phi_t(x)\phi_t(y)\,a_x^* a_y^*+\mathrm{h.c.}\big\} \Big)^2 \\
& \; + \Big( \frac{1}{\,\sqrt{N}\,}\int\rd x \,\rd y\, \frac{1}{|x-y|+\alpha_N} \big\{\phi_t(y) \,a_x^* a_y^* \,{\bf 1}_{\vartheta N} (\cN)\,a_x+\mathrm{h.c.}\big\} \Big)^2 \, .
\end{split}
\end{equation} 

To bound the second term on the r.h.s. of the last equation, we note that
\[ \int\rd x\,\left(\frac{1}{|.| +\alpha_N} * |\phi_t|^2\right) \,a_x^* a_x \leq \sup_x \left( \frac{1}{|.|}*|\phi_t|^2 \right) \, \cN \lesssim \| \phi_t \|_{H^{1/2}}^2 \, \cN \, . \]
Since moreover $\cN$ commutes with the operator on the l.h.s., we conclude that
\begin{equation}\label{eq:K2t-1} \left( \int\rd x\,\left(\frac{1}{|.| +\alpha_N} * |\phi_t|^2\right) \,a_x^* a_x \right)^2 \lesssim \cN^2 \,. \end{equation}
Analogously, the third term on the r.h.s. of (\ref{K^2=Ltilde^2+split_0}) is bounded by
\begin{equation}\label{eq:K2t-2} \Big( \int\rd x \,\rd y\, \frac{1}{|x-y|+\alpha_N} \,\overline{\phi}_t(x) \,\phi_t(y)\,a_y^* a_x \Big)^2 \lesssim \cN^2 \, . \end{equation}
Next, the terms on the third line of (\ref{K^2=Ltilde^2+split_0}) can be controlled as follows. Let 
\[ A = \int \rd x \rd y \, \frac{1}{|x-y|+\alpha_N} \phi_t (x) \phi_t (y) a_x^* a_y^* \, .\]
Since $(A+A^*)^2 \leq 2 (AA^* +A^* A)$, we find
\[ \begin{split}
\langle \psi, (A+A^*)^2 \psi \rangle \lesssim \; & \int \rd x \rd y \rd x' \rd y' \, \frac{\phi_t (x) \phi_t (y)}{|x-y|+\alpha_N}\, \frac{\overline{\phi}_t (x') \, \overline{\phi}_t (y')}{|x'-y'|+\alpha_N}  \, \langle \psi, a_x^* a_y^* a_{x'} a_{y'} \psi \rangle \\ &+  \int \rd x \rd y \rd x' \rd y' \, \frac{\phi_t (x) \phi_t (y)}{|x-y|+\alpha_N}\, \frac{\overline{\phi}_t (x') \, \overline{\phi}_t (y')}{|x'-y'|+\alpha_N}  \,  \langle \psi , [ a_x^* a_y^* , a_{x'} a_{y'} ] \psi \rangle  \\
 \lesssim \; & \int \rd x \rd y \rd x' \rd y' \, \frac{|\phi_t (x)| |\phi_t (y)|}{|x-y|+\alpha_N}\, \frac{|\phi_t (x')| |\phi_t (y')|}{|x'-y'|+\alpha_N} \,\| a_x a_y \, \psi \| \, \| a_{x'} a_{y'} \psi \| \\ &+  \int \rd x \rd y \rd x' \, |\phi_t (x)|^2 \frac{|\phi_t (y)|}{|x-y|+\alpha_N}\, \frac{|\phi_t (x')|}{|x-x'|+\alpha_N} \, \| a_y \, \psi \| \, \| a_{x'} \psi \| 
 \\ &+ \int \rd x \rd y \, \frac{|\phi_t (x)|^2 |\phi_t (y)|^2}{(|x-y|+\alpha_N)^2} \, \| \psi \|^2
 \end{split}\]
for arbitrary $\psi \in \cF$. Here we used that
\begin{equation}\label{eq:comm1} \begin{split} [a_x a_y, a_{x'}^* a_{y'}^* ] =  \, &a_{y'}^* a_x \delta (y-x') + 
a_{x'}^* a_x \delta (y - y') + a_{y'}^* a_y  \delta (x-x')+  a_{x'}^* a_y \delta (x-y') 
\\ &+  \delta (x-x') \delta (y-y') + \delta (y'-x)\delta (y-x'). \end{split} \end{equation}
With Schwarz inequality, we obtain
\[ \begin{split} 
\langle \psi, (A+A^*)^2 \psi \rangle \lesssim \;   & \int \rd x \rd y \rd x' \rd y' \, \frac{|\phi_t (x)|^2 |\phi_t (y)|^2}{(|x-y|+\alpha_N)^2} \, \| a_{x'} a_{y'} \psi \|^2 \\ &+  \int \rd x \rd y \rd x' \,\frac{|\phi_t (x)|^2 \, |\phi_t (y)|^2}{(|x-y|+\alpha_N)^2} \, \| a_{x'} \psi \|^2 + \int \rd x \rd y \, \frac{|\phi_t (x)|^2 |\phi_t (y)|^2}{(|x-y|+\alpha_N)^2} \, \| \psi \|^2
 \\ \lesssim \; & \| \phi_t \|_{H^1}^2 \, \| \phi_t \|^2 \,  \langle \psi, (\cN+1)^2 \psi \rangle \,. \end{split} \] 
Thus
\begin{equation}\label{eq:K2t-4} \Big( \frac{1}{2}\int\rd x \,\rd y\,\frac{1}{|x-y|+\alpha_N} \big\{\phi_t(x)\phi_t(y)\,a_x^* a_y^*+\mathrm{h.c.}\big\} \Big)^2 \lesssim \, (\cN+1)^2 \, .\end{equation}
Now, we estimate the terms on the fourth line of (\ref{K^2=Ltilde^2+split_0}). Let
\[ B = \frac{1}{\sqrt{N}} \int \rd x \rd y \, \frac{1}{|x-y|+\alpha_N} \, a_x^* a_y^* \phi_t (y) \, {\bf 1}_{\vartheta N} (\cN) a_x \]
Then $(B+B^*)^2 \lesssim BB^* + B^* B$. The term $BB^*$ can be bounded by
\begin{equation}\label{B4B4*}
\begin{split}
\big\la \psi, B B^* \psi\big\ra\; & =\;\frac{1}{N}\int\rd x\,\rd y\,\rd x'\rd y'\,\frac{1}{|x-y|+\alpha_N} \frac{1}{|x'-y'|+\alpha_N} \phi_t (y)\,\overline{\phi}_t (y') \\
& \qquad \qquad \times\big\la{\bf 1}_{\vartheta N} (\cN-2)\,\psi ,a_x^* a_y^* a_x a_{x'}^* a_{x'} a_{y'}\, {\bf 1}_{\vartheta N} (\cN-2)\,\psi\big\ra \\
& =\;\frac{1}{N}\int\rd x\,\rd y\,\rd x'\rd y'\, ,\frac{1}{|x-y|+\alpha_N} \frac{1}{|x'-y'|+\alpha_N}  \phi_t(y)\,\overline{\phi}_t(y') \\
& \qquad \qquad \times\big\la{\bf 1}_{\vartheta N} (\cN-2)\,\psi,a_x^* a_y^* a_{x'}^* a_{x'} a_x  a_{y'}\, {\bf 1}_{\vartheta N} (\cN-2)\,\psi\big\ra \\
& \quad  + \frac{1}{N}\int\rd x\,\rd y\,\rd y'\, \frac{1}{|x-y|+\alpha_N} \frac{1}{|x-y'|+\alpha_N} \phi_t(y)\,\overline{\phi}_t(y') \\
& \qquad \qquad \times\big\la{\bf 1}_{\vartheta N} (\cN-2)\,\psi,a_x^* a_y^*  a_x    a_{y'}\, {\bf 1}_{\vartheta N} (\cN-2)\,\psi\big\ra \\
\end{split}\end{equation}
for every $\psi \in \cF$. {F}rom Schwarz inequality, we find
\begin{equation}
\begin{split}
 \big\la \psi, B B^* \psi\big\ra\;  &\lesssim
\;\frac{1}{N}\int\rd x\,\rd y\,\rd x'\rd y'\, ,\frac{|\phi_t (y)|^2}{(|x-y|+\alpha_N)^2} \, \| a_{x'} a_x a_{y'} \, {\bf 1}_{\vartheta N} (\cN-2)\,\psi \|^2 \\ &\quad+ \frac{1}{N}\int\rd x\,\rd y\,\rd y'\, \frac{|\phi_t (y)|^2}{(|x-y|+\alpha_N)^2}\,
\| a_x a_{y'} \, {\bf 1}_{\vartheta N} (\cN-2)\,\psi \|^2  \\
&\lesssim \; \frac{1}{N} \,\left(\sup_x \int \rd y \frac{|\phi_t (y)|^2}{|x-y|^2} \right) \, 
\int\rd x \,\rd x'\rd y'\, \| a_{x'} a_x a_{y'}\,  {\bf 1}_{\vartheta N} (\cN-2)\,\psi \|^2 \\ 
& \quad  + \frac{1}{N} \left(\sup_x \int \rd y \frac{|\phi_t (y)|^2}{|x-y|^2}\right)
\int\rd x\, \rd y'\, \| a_x a_{y'} \, {\bf 1}_{\vartheta N} (\cN-2)\,\psi \|^2 \\
&\lesssim 
\; \frac{1}{N} \, \| \phi_t \|_{H^1}^2 \,  \| (\cN+1)^{3/2} \,{\bf 1}_{\vartheta N} (\cN-2) \, \psi \|^2 \,  .
\end{split}\end{equation}
The term $B^*B$, on the other hand, is given by
\begin{equation}\label{B4*B4}
\begin{split}
 \big\la \psi , B^* B \psi \big\ra\; & =\;\frac{1}{N}\int\rd x\,\rd y\,\rd x'\rd y'\,\frac{1}{|x-y|+\alpha_N} \frac{1}{|x'-y'| +\alpha_N} \,\overline{\phi}_t(y) \,\phi_t(y') \\
& \qquad \qquad \times\big\la{\bf 1}_{\vartheta N} (\cN-1)\,\psi ,a_x^*  a_x a_y a_{x'}^* a_{y'}^* a_{x'} \, {\bf 1}_{\vartheta N} (\cN-1)\,\psi\big\ra \\
& =\;\frac{1}{N}\int\rd x\,\rd y\,\rd x'\rd y'\,\frac{1}{|x-y|+\alpha_N} \frac{1}{|x'-y'| +\alpha_N} \,\overline{\phi}_t(y) \,\phi_t(y') \\
& \qquad \qquad \times\big\la{\bf 1}_{\vartheta N} (\cN-1)\,\psi,a_x^* a_{x'}^* a_{y'}^* a_x a_{x'} a_y \,  {\bf 1}_{\vartheta N} (\cN-1)\,\psi\big\ra \\
& \quad + \frac{1}{N}\int\rd x\,\rd y\,\rd x' \, \rd y'\, \frac{1}{|x-y|+\alpha_N} \frac{1}{|x'-y'| +\alpha_N} \,\overline{\phi}_t(y)\,\phi_t(y') \\
& \qquad \qquad \times\big\la{\bf 1}_{\vartheta N}(\cN-1)\,\psi,a_x^* \left[ a_x a_y , a_{x'}^* a_{y'}^* \right]
\, a_{x'} \, {\bf 1}_{\vartheta N} (\cN-1)\,\psi\big\ra 
\end{split}
\end{equation}
The first term on the r.h.s. is bounded in absolute value by 
\begin{equation}\label{eq:B4-term1} \begin{split} 
\frac{1}{N}\int\rd x\,\rd y\,\rd x'\rd y'\,&\frac{|\phi_t (y)|}{|x-y|} \frac{|\phi_t (y')|}{|x'-y'|} \,
\| a_x a_{x'} a_{y'} \, {\bf 1}_{\vartheta N} (\cN-1)\,\psi \| \| a_x a_{x'} a_y  \, {\bf 1}_{\vartheta N} (\cN-1)\,\psi \|  \\ \lesssim \; & \frac{1}{N} \, \left(\sup_x \, \int \rd y \, \frac{|\phi_t (y)|^2}{|x-y|^2}  \right) \int\rd x\,\rd x' \, \rd y' \, \| a_x a_{x'} a_{y'} \, {\bf 1}_{\vartheta N} (\cN-1)\,\psi \|^2 \\ \lesssim \; & \frac{1}{N} \, \| \phi_t \|_{H^1}^2 \, \| (\cN+1)^{3/2} \, {\bf 1}_{\vartheta N} (\cN-1) \psi \|^2  \, .
\end{split}\end{equation}
 When we insert (\ref{eq:comm1}) in the second term on the r.h.s. of (\ref{B4*B4}), we obtain contributions quartic in the creation and annihilation operators of the form 
\[ \frac{1}{N} \int \rd x \, \rd y \, \rd y' \, \frac{1}{|x-y|+\alpha_N} \frac{1}{|y-y'|+\alpha_N} \overline{\phi}_t (y) \phi_t (y') \, \langle {\bf 1}_{\vartheta N} (\cN-1) \psi, a_x^* a_{y'}^* a_x a_{y} \, {\bf 1}_{\vartheta N} (\cN-1) \psi \rangle \]
whose absolute value can be bounded by
\[ \begin{split} 
\Big| \frac{1}{N} \int \rd x \, &\rd y \, \rd y'  \,  \frac{1}{|x-y|+\alpha_N} \frac{1}{|y-y'|+\alpha_N|} \overline{\phi}_t (y) \phi_t (y') \, \langle {\bf 1}_{\vartheta N} (\cN-1) \psi, a_x^* a_{y'}^* a_x a_{x'} \, {\bf 1}_{\vartheta N} (\cN-1) \psi \rangle \Big| \\ \leq  \; & \frac{1}{N} \int \rd x \, \rd y \, \rd y'  \, \frac{1}{|x-y|} \frac{1}{|y-y'|} |\phi_t (y)| \,  |\phi_t (y')| \, \| a_{y'} a_x\, {\bf 1}_{\vartheta N} (\cN-1) \psi \| \, \| a_x a_{x'} \, {\bf 1}_{\vartheta N} (\cN-1) \psi \| \\ \leq \; & 
\frac{1}{N} \left( \sup_y \int \rd y' \, \frac{1}{|y-y'|^2} \,  |\phi_t (y')|^2 \right)  \int \rd x \, \rd y \,    \, \| a_x a_{y} \,{\bf 1}_{\vartheta N} (\cN-1) \psi \|^2 \\ &+ \frac{1}{N} \left( \sup_x \int \rd y \frac{1}{|x-y^2|} |\phi_t (y)|^2\right) \int \rd x \, \rd y'  \,  \,  \, \| a_{y'} a_x \, {\bf 1}_{\vartheta N} (\cN-1) \psi \|^2 \\ \lesssim \;& \frac{1}{N} \, \| \phi_t \|_{H^1}^2 \, \| (\cN +1)\,{\bf 1}_{\vartheta N} (\cN-1) \psi \|^2 \,.\end{split}
\] 
The other terms arising when we insert (\ref{eq:comm1}) in the second summand on the r.h.s. of (\ref{B4*B4}) (both the quartic and the quadratic terms) can be bounded analogously. 
Together with (\ref{eq:B4-term1}), we conclude that
\begin{equation}\label{eq:K2t-5}
\big\la \psi , B^* B \psi \big\ra\; \lesssim \; \frac{1}{N} \, \| \phi_t \|_{H^1}^2 \,  \| (\cN+1)^{3/2} \, {\bf 1}_{\vartheta N} (\cN-1) \psi \|^2 \, .
\end{equation}
{F}rom (\ref{B4B4*}) and (\ref{eq:K2t-5}), we find 
\begin{equation}\label{eq:K2t-6} \Big( \frac{1}{\,\sqrt{N}\,}\int\rd x \,\rd y\, \frac{1}{|x-y|+\alpha_N} \big\{\phi_t(y) \,a_x^* a_y^* \,{\bf 1}_{\vartheta N} (\cN)\,a_x+\mathrm{h.c.}\big\} \Big)^2 \lesssim \frac{1}{N}  (\cN+1)^{3} \, {\bf 1}_{\vartheta N} (\cN-2) \, . 
\end{equation}
Combining (\ref{eq:K2t-1}),  (\ref{eq:K2t-2}),  (\ref{eq:K2t-4}), and (\ref{eq:K2t-6}) we conclude (since ${\bf 1}_{\vartheta N} \leq 1$) that 
\begin{equation}\label{eq:cKles} \cK^2  \lesssim \wt{\cM}_N^2 (t) + (\cN+1)^3 \, . \end{equation}

Next, we observe that there exists a constant $C=C(T,\kappa, \| \ph\|_{H^{2}})$ such that 
\begin{equation}\label{eq:bdNwtU} \langle \wt{\cW}_N (t;0) \Omega, (\cN+1)^3 \, \wt{\cW}_N (t;0) \Omega \rangle \leq C \end{equation} for all $|t| \leq T$. The proof of this bound is analogous to the proof of Proposition \ref{prop:ntildebd} (see Lemma 3.5, with $M = \vartheta N$, and its proof in \cite{RS}). The only difference is that the generator $\wt{\cM}_N (t)$ contains a cutoff also in the quartic term (while in Proposition \ref{prop:ntildebd}, the cutoff appeared only in the cubic term of the generator $\cM_N (t)$); this difference does not play any role in the proof of (\ref{eq:bdNwtU}) because the quartic term (with or without cutoff) commutes with the number of particles operator $\cN$ (and thus with its powers). 

\medskip

Finally, we control the growth of the expectation of $\wt{\cM}_N^2 (t)$. To this end we compute, using (\ref{eq:cUNtilde}), 
\[ \frac{\rd}{\rd t} \left\langle \wt{\cW}_N (t;0) \Omega, \wt{\cM}_N^2 (t) \wt{\cW}_N (t;0) \Omega \right\rangle = \left\langle \wt{\cW}_N (t;0) \Omega, \left( \wt{\cM}_N (t) \dot{\wt{\cM}}_N (t)+ \dot{\wt{\cM}}_N (t) \, \wt{\cM}_N (t) \right) \wt{\cW}_N (t;0)\Omega \right\rangle \]
and thus 
\begin{equation}\label{eq:dotL} \left| \frac{\rd}{\rd t} \left\langle \wt{\cW}_N (t;0) \Omega, \wt{\cM}_N^2 (t) \, \wt{\cW}_N (t;0) \Omega \right\rangle^{1/2} \right| \leq \, \left\langle \wt{\cW}_N (t;0) \Omega, \dot{\wt{\cM}}^2_N (t) \wt{\cW}_N (t;0)\Omega \right\rangle^{1/2}\,. \end{equation}
We have
\begin{equation}\label{eq:cLNtilde}
\begin{split}
\dot{\widetilde{\mathcal{M}}}_N(t)\;&=- \; \lambda \int\rd x\left(\frac{1}{|.| +\alpha_N}*(\dot{\overline{\phi}}_t \phi_t+\overline{\phi}_t \dot{\phi}_t)\right)(x)\,a_x^* 
a_x \\
& \quad- \lambda \int \rd x\,\rd y\,\frac{1}{|x-y| +\alpha_N} \, \left(\dot{\overline{\phi}}_t (x)\phi_t(y) + \overline{\phi}_t (x) \dot{\phi}_t (y) \right) \,a_y^* a_x \\
& \quad- \lambda \int \rd x\,\rd y\,\frac{1}{|x-y| +\alpha_N} \, \left( \dot{\phi}_t(x)\phi_t(y) \,a_y^*a_x^* + \mathrm{h.c.} \right) \\
& \quad- \frac{\lambda}{\,\sqrt{N}\,}\int \rd x\,\rd y\,\frac{1}{|x-y|+\alpha_N} \, \left(\dot{\phi}_t(y) a_x^* a_y^* \, {\bf 1}_{\vartheta N} (\cN) \, a_x+\mathrm{h.c.} \right) \, .
\end{split}
\end{equation}
Observe that from the (regularized) Hartree equation (\ref{eq:reg-aN}) 
we easily find that $\| \dot{\phi}_t \| \lesssim \| \phi_t \|_{H^1}$ and
\begin{equation}
\begin{split}
\big \|\nabla\dot{\phi}_t\big\| \;&\lesssim\;\big\|(1-\Delta)\phi_t\big\|_2+\Big\|\nabla\Big(\frac{1}{|\cdot|+\alpha}*|\phi_t|^2\Big)\phi_t\Big\|_2 \\
& \lesssim \;\|\phi_t\|_{H^2}+\Big\|\frac{1}{|\cdot|^2}* |\phi_t|^2 \Big\|_3\,\|\phi_t\|_6\;+\;\Big\|\frac{1}{|\cdot|}*|\phi_t|^2\Big\|_\infty\,\big\|\nabla\phi_t\big\|_2 \\
& \lesssim \;\|\phi_t\|_{H^2} +\|\phi_t\|_{H^{1/2}}^2\,\|\phi_t\|_{H^1} \, .
\end{split}
\end{equation} 
This implies, by (\ref{eq:bdpha}), that there exists a constant $C = C(T,\kappa, \| \ph \|_{H^2})$ such that \begin{equation}
\label{eq:dotphi}
\| \dot{\phi}_t \|_{H^1} \leq C \qquad \text{for all $t \in \bR$ with $|t| \leq T$.} \end{equation}

Next, we bound the square of the terms on the r.h.s. of (\ref{eq:cLNtilde}). Similarly to (\ref{eq:K2t-1}) and (\ref{eq:K2t-2}), we find
\begin{equation}\label{eq:L-1} \left( \int\rd x \, \left(\frac{1}{|.| +\alpha_N}*(\dot{\overline{\phi}}_t \phi_t+\overline{\phi}_t\dot{\phi}_t)\right)(x)\,a_x^* a_x \right)^2 \lesssim \| \phi_t \|_{H^1} \| \dot{\phi}_t \|_{H^1} \cN^2 \end{equation}
and 
\begin{equation}\label{eq:L-2} \left( \int \rd x\,\rd y\,\frac{1}{|x-y| +\alpha_N} \, \left(\dot{\overline{\phi}}_t (x)\phi_t(y) + \overline{\phi}_t (x) \dot{\phi}_t (y)\right)  \,a_y^* a_x \right)^2 \leq  \| \phi_t \|_{H^1} \| \dot{\phi}_t \|_{H^1} \cN^2 \, .\end{equation}
Moreover, similarly to (\ref{eq:K2t-4}), we obtain
\begin{equation}\label{eq:L-3} \left( \int \rd x\,\rd y\,\frac{1}{|x-y| +\alpha_N} \, \left(\dot{\phi}_t(x)\phi_t(y)+\phi_t(x) \dot{\phi}_t(y)\right)\,a_y^*a_x^*+\mathrm{h.c.} \right)^2 \lesssim \| \phi_t \|_{H^1} \| \dot{\phi} \|  (\cN+1)^2 \, .\end{equation}
Finally, analogously to (\ref{eq:K2t-5}) (replacing $\phi$ with $\dot{\phi}$) we have
\begin{equation}\label{eq:L-4} \left( \frac{1}{\,\sqrt{N}\,}\int \rd x\,\rd y\,\frac{1}{|x-y|+\alpha_N} \, \dot{\phi}_t(y) a_x^* a_y^* \, {\bf 1}_{\vartheta N} (\cN)\,a_x+\mathrm{h.c.} \right)^2 \lesssim \frac{1}{N} \| \dot{\phi}_t \|_{H^1}^2 \, (\cN+1)^3  \, .\end{equation}
{F}rom (\ref{eq:cLNtilde}), (\ref{eq:L-1}),  (\ref{eq:L-2}),  (\ref{eq:L-3}),  (\ref{eq:L-4}), we find, using (\ref{eq:bdNwtU}), that 
\[ \langle \wt{\cW}_N (t;0) \Omega, \dot{\wt{\cM}}^2_N (t) \wt{\cW}_N (t;0)\Omega \rangle \lesssim 1. \]
Eq. (\ref{eq:dotL}) then implies that there exists $C=C(\kappa,T,\| \ph \|_{H^2})$ such that
\[ \langle \wt{\cW}_N (t;0) \Omega, \wt{\cM}_N^2 (t) \, \wt{\cW}_N (t;0) \Omega \rangle \leq C \] for all $t \in \bR$ with $|t| \leq T$. Proposition \ref{prop:K2bdt} now follows from (\ref{eq:cKles}) and (\ref{eq:bdNwtU}). \qed

\bigskip

\begin{lemma}\label{Lemma:VtK}
There exists a universal constant $C>0$ such that
\begin{equation}
\begin{split}
\left(\int \rd x \rd y \, \frac{1}{|x-y|+\alpha} \, a_x^* a_y^* \, {\bf 1}_{\vartheta N} \, (\cN) a_y a_x\right)^2 \leq \; & C \,\vartheta^2\,\cK^2\,
\end{split}
\end{equation}
for all $\alpha, \vartheta >0$. 
\end{lemma} 

\begin{proof}
Denote
\[ \wt\cV = \int \rd x \rd y \, \frac{1}{|x-y|+\alpha} \, a_x^* a_y^* \,  {\bf 1}_{\vartheta N} (\cN) a_y a_x \, . \]
Then $\wt{\cV}$ (and thus $\wt{\cV}^2$) leaves the number of particles invariant and, on the $n$-particle sector, we have
\begin{equation}
(\wt{\cV}^2)^{(n)} = \bigg(\frac{1}{N}\sum_{1\leq i<j\leq n}\frac{1}{\;|x_i-x_j|+\alpha} \bigg)^2\qquad \qquad \text{if } n\leq \vartheta N
\end{equation} 
and $(\wt{\cV}^2)^{(n)}=0$ when $n>\vartheta N$. Using the operator inequality (see, for example, Lemma 9.1 in \cite{ES}) \[ \frac{1}{|x-y|^2} \lesssim (1-\Delta_x)^{1/2} \, (1-\Delta_y)^{1/2} \] we find 
\begin{equation}
\begin{split}
(\wt{\cV}^2 )^{(n)} \;& \lesssim \;\frac{n^2}{N^2}\sum_{1\leq i<j\leq n}\frac{1}{\big(|x_i-x_j|+\alpha)^2} \lesssim \; \frac{n^2}{N^2}\sum_{1\leq i<j\leq n}(1-\Delta_{x_i})^{1/2}(1-\Delta_{x_j})^{1/2} \\
&\lesssim \; \frac{n^2}{N^2}\bigg(\sum_{j=1}^n(1-\Delta_{x_j})^{1/2}\bigg)^2 \lesssim \;\vartheta^2 (\cK^2)^{(n)}
\end{split} 
\end{equation} 
and the lemma is proven.
\end{proof}

\subsection{Weak bounds on growth of $\cK^2$ with respect to fluctuation dynamics}
\label{sec:K2bd}

In this subsection, we show Proposition \ref{prop:K2bd}. Again, we will need the estimate (\ref{eq:bdpha}); recall also that in this section we use the shorthand notation $\phi_t \equiv \ph_t^{(\alpha_N)}$ for the solution of (\ref{eq:reg-aN}).

\bigskip

We write
\[ \begin{split} \big\langle \cU_N (t;0) \Omega, \cK^2  &\cU_N (t;0) \Omega \big\rangle \\ = \; &\int \rd x \rd y \, \left\langle  \cU_N (t;0) \Omega, (1-\Delta_x)^{1/4} \, a_x^* \, (1-\Delta_x)^{1/4} \, a_x \right.\\ &\left. \hspace{.5cm} \times (1-\Delta_y)^{1/4} \, a_y^* \, (1-\Delta_y)^{1/4} \, a_y \, \cU_N (t;0) \Omega \right\rangle  \\ =\; &  \int \rd x \rd y \, \left\langle \Omega, (1-\Delta_x)^{1/4} \cU^*_N (t;0) \, a_x^* \, \cU_N (t;0)\, (1-\Delta_x)^{1/4} \, \cU^*_N (t;0) \, a_x \,  \cU_N (t;0) \right. \\ &\left. \hspace{.5cm}  \times \, (1-\Delta_y)^{1/4} \, \cU^*_N (t;0) \, a_y^* \, \cU_N (t;0) \, (1-\Delta_y)^{1/4} \, \cU^*_N (t;0) \, a_y \, \cU_N (t;0) \Omega \right\rangle \,.
\end{split}
\]
Next we use that (see (\ref{eq:def-U})) 
\[ \cU^*_N (t;0) \, a_x \, \cU_N (t;0) = W^* (\sqrt{N} \ph) e^{i\cH_N^\alpha t} (a_x - \sqrt{N} \phi_t (x) ) e^{-i\cH^\alpha_N t} W (\sqrt{N} \ph) \] 
to conclude that
\begin{equation}\label{eq:K2-1} \begin{split}
\big\langle \cU_N (t;0) &\Omega, \cK^2  \cU_N (t;0) \Omega \big\rangle \\ = \; & \int \rd x \rd y \, \left\langle e^{-i \cH_N^\alpha t} \, W (\sqrt{N} \ph) \Omega, (1-\Delta_x)^{1/4} (a_x^* -\sqrt{N} \overline{\phi}_t (x)) \,  (1-\Delta_x)^{1/4} (a_x -\sqrt{N} \phi_t (x) ) \right. \\ & \left. \times  (1-\Delta_y)^{1/4} \, (a_y^* - \sqrt{N} \overline{\phi}_t (y)) \, (1-\Delta_y)^{1/4} \, (a_y  - \sqrt{N} \phi_t (y)) \, e^{-i \cH_N^\alpha t} W (\sqrt{N} \ph) \Omega \right\rangle \,.
\end{split}
\end{equation}
For $f \in L^2 (\bR^3)$, let \[ \pi (f) = a^* (f) + a (f) = \int \rd x \, \left( f(x) \, a_x^* + \overline{f}(x) \, a_x \right) . \]
Then, from (\ref{eq:K2-1}), we obtain 
\[ \begin{split}
\big\langle \cU_N (t;0) &\Omega, \cK^2  \cU_N (t;0) \Omega \big\rangle \\ = \; &\big\langle e^{-i \cH_N^\alpha t} \, W(\sqrt{N} \ph) \Omega , \cK^2 \, e^{-i\cH_N^\alpha t} W (\sqrt{N} \ph) \Omega \big\rangle \\ &+ 2 \sqrt{N} \text{Re } \big\langle e^{-i\cH_N^\alpha t} \, W(\sqrt{N} \ph) \Omega , \cK \, \pi ((1-\Delta)^{1/4} \phi_t) \, e^{-i \cH_N^\alpha t} W (\sqrt{N} \ph) \Omega \big\rangle  \\
 &+ 2 N \, \| \phi_t \|_{H^{1/2}}^2 \big\langle e^{-i \cH_N^\alpha t} \, W(\sqrt{N} \ph) \Omega ,  \cK \, e^{-i\cH_N^\alpha t} W (\sqrt{N} \ph) \Omega \big\rangle  \\
 &+ N \big\langle e^{-i \cH_N^\alpha t} \, W(\sqrt{N} \ph) \Omega ,  \pi^2 ((1-\Delta)^{1/4} \phi_t)  \, e^{-i\cH_N^\alpha t} W (\sqrt{N} \ph) \Omega \big\rangle  \\
 &+ 2 N^{3/2}  \| \phi_t \|_{H^{1/2}}^2 \, \big\langle e^{-i \cH_N^\alpha t} \, W(\sqrt{N} \ph) \Omega ,  \pi ((1-\Delta)^{1/4} \phi_t) \, e^{-i\cH_N^\alpha t} W (\sqrt{N} \ph) \Omega \big\rangle  \\
 &+ N^2 \| \phi_t \|_{H^{1/2}}^4 \,.
 \end{split}
 \]
 Using Schwarz inequality, we find
 \begin{equation}\label{eq:K2-2} \begin{split}
\big\langle \cU_N (t;0) \Omega, \cK^2  \cU_N (t;0) \Omega \big\rangle \lesssim \; &\big\langle e^{-i\cH_N^\alpha t} \, W(\sqrt{N} \ph) \Omega , \cK^2 \, e^{-i\cH_N^\alpha t} W (\sqrt{N} \ph) \Omega \big\rangle \\
 &+ N \big\langle e^{-i \cH_N^\alpha t} \, W(\sqrt{N} \ph) \Omega ,  \pi^2 ((1-\Delta)^{1/4} \phi_t)  \, e^{-i\cH_N^\alpha t} W (\sqrt{N} \ph) \Omega \big\rangle  \\
 &+ N^2 \| \phi_t \|_{H^{1/2}}^4 \,.
\end{split}
\end{equation}
Next, we observe that, for arbitrary $f \in L^2 (\bR^3)$ and $\psi \in \cF$, 
\begin{equation}\label{eq:K2-3} \langle \psi, \pi^2 (f) \psi \rangle = \| \pi (f) \psi \|^2 \lesssim \| a (f) \psi \|^2 + \| a^* (f) \psi \|^2 \lesssim \| f \|^2 \langle \psi, (\cN +1) \psi \rangle \, . \end{equation}
Moreover, we have
\[ \cH_N^{\alpha} = \cK - \cV, \qquad \text{where} \qquad  \cV = \frac{\lambda}{2N} \int \rd x \rd y \, \frac{1}{|x-y|+\alpha} \, \a_x^* a_y^* a_y a_x \, . \]
Since 
\[ \cV \lesssim \frac{1}{N\alpha} \cN^2 \] and since $[\cV, \cN]=0$, we conclude that 
\begin{equation}\label{eq:K2-4} \cK^2 \lesssim (\cH_N^{(\alpha)})^2 + \cV^2 \lesssim (\cH_N^{(\alpha)})^2 + \frac{1}{N^2 \alpha^2} \cN^4 \end{equation}
Inserting (\ref{eq:K2-3}) and (\ref{eq:K2-4}) in (\ref{eq:K2-2}), we find
 \begin{equation}\label{eq:K2-5} \begin{split}
\big\langle \cU_N (t;0) \Omega, &\cK^2  \cU_N (t;0) \Omega \big\rangle \\ \lesssim \; &\big\langle  W(\sqrt{N} \ph) \Omega , (\cH_N^{\alpha})^2 \, W (\sqrt{N} \ph) \Omega \big\rangle + \frac{1}{N^2 \alpha^2} \big\langle W(\sqrt{N} \ph) \Omega , \cN^4 \, W (\sqrt{N} \ph) \Omega \big\rangle \\
 &+ N \| \phi_t \|^2_{H^{1/2}} \, \big\langle W(\sqrt{N} \ph) \Omega ,  (\cN+1)  \, W (\sqrt{N} \ph) \Omega \big\rangle  + N^2 \| \phi_t \|_{H^{1/2}}^4 \\
  \lesssim \; &\big\langle W(\sqrt{N} \ph) \Omega , \cK^2 \, W (\sqrt{N} \ph) \Omega \big\rangle + \frac{1}{N^2 \alpha^2} \big\langle W(\sqrt{N} \ph) \Omega , \cN^4 \, W (\sqrt{N} \ph) \Omega \big\rangle \\
 &+ N \| \phi_t \|^2_{H^{1/2}} \, \big\langle W(\sqrt{N} \ph) \Omega ,  (\cN+1)  \, W (\sqrt{N} \ph) \Omega \big\rangle  + N^2 \| \phi_t \|_{H^{1/2}}^4 \,.
\end{split}
\end{equation}
Using the properties of Weyl operators listed in Lemma \ref{lm:coh}, it is simple to check that
\begin{equation}\label{eq:K2-6} \big\langle W(\sqrt{N} \ph) \Omega ,  (\cN+1)  \, W (\sqrt{N} \ph) \Omega \big\rangle  \lesssim N \quad \text{and } \quad  \big\langle W(\sqrt{N} \ph) \Omega , \cN^4 \, W (\sqrt{N} \ph) \Omega \big\rangle \lesssim N^4 \,. \end{equation} 
Moreover, we have
\begin{equation}\label{eq:K2-8} \begin{split}
\big\langle W(\sqrt{N} \ph) & \Omega , \cK^2 \, W (\sqrt{N} \ph) \Omega \big\rangle \\ = \; & \int \rd x \rd y \, \big\langle W(\sqrt{N} \ph) \Omega , (1-\Delta_x)^{1/4} a_x^* \, (1-\Delta_x)^{1/4} a_x \\ &\hspace{1cm} \times  (1-\Delta_y)^{1/4} a_y^* \, (1-\Delta_y)^{1/4} a_y \, W (\sqrt{N} \ph) \Omega \big\rangle \\ = \; &  \int \rd x \rd y \, \big\langle \Omega , (1-\Delta_x)^{1/4} (a_x^* -\sqrt{N} \overline{\ph} (x)) \, (1-\Delta_x)^{1/4} (a_x -\sqrt{N} \ph (x)) \\ &\hspace{1cm} \times  (1-\Delta_y)^{1/4} (a_y^* - \sqrt{N} \overline{\ph} (y)) \, (1-\Delta_y)^{1/4} (a_y -\sqrt{N} \ph (y)) \, W (\sqrt{N} \ph) \Omega \big\rangle \\= \; & N^2 \| \ph \|_{H^{1/2}}^4 + N \| \ph \|_{H^1}^2 
\end{split}\end{equation}
Inserting (\ref{eq:K2-6}) and (\ref{eq:K2-8}) into (\ref{eq:K2-5}), and using the bound (\ref{eq:bdpha}), we conclude that
\[  \big\langle \cU_N (t;0) \Omega, \cK^2  \cU_N (t;0) \Omega \big\rangle \lesssim N^2 + \frac{N^2}{\alpha^2} \,. \]
This completes the proof of Proposition \ref{prop:K2bd}.\qed

\subsection{Comparison of fluctuation dynamics with regularized dynamics}
\label{sec:U-Utilde}

In this section, we prove Proposition \ref{prop:U-Utilde}.

\bigskip

We rewrite
\begin{equation}\label{cook_U-Utilde}
\begin{split}
\big\|(\UU-\tUU)\Omega\big\|\;&=\;\big\|(1-\UUc\,\tUU)\Omega\big\| \\
& \leq\;\int_0^t\big\|\mathcal{U}^*_N(s;0)\big(\cL_N(s)-\widetilde{\cM}_N(s)\big)\,\widetilde{\mathcal{W}}_N(s;0)\,\Omega\big\|\,\rd s \\
& \leq\;\int_0^t\big\|\big(\cL_N(s)-\widetilde{\cM}_N(s)\big)\,\widetilde{\mathcal{W}}_N(s;0)\,\Omega\big\|\,\rd s\,.
\end{split}
\end{equation} 
We recall that
\begin{equation}\label{eq:diffe}
\begin{split}
\tLL-\LL\;&=\;\frac{\lambda}{\,\sqrt{N}\,}\int\rd x \,\rd y\,\frac{1}{|x-y|+\alpha_N} \, \phi_t(y) \,a_x^* a_y^* \,(1-{\bf 1}_{\vartheta N} (\cN))\,a_x+\mathrm{h.c.} \\
& \quad+\frac{\lambda}{\,2N}\int\rd x \,\rd y\,\frac{1}{|x-y|+\alpha_N}  \,a_x^* a_y^* \,(1-{\bf 1}_{\vartheta N} (\cN))\,a_y a_x\,.
\end{split}
\end{equation} 
Analogously to (\ref{eq:K2t-6}), but with ${\bf 1}_{\vartheta N}$ replaced by $1-{\bf 1}_{\vartheta N}$, the square of the terms on the first line of last equation can be bounded by 
\begin{equation}\label{eq:cubic}
\begin{split}
\Big( \frac{1}{\,\sqrt{N}\,}\int\rd x \,\rd y\,\frac{1}{|x-y|+\alpha_N} \, \phi_t(y) \,a_x^* a_y^* \,(1-{\bf 1}_{\vartheta N} (\cN))\, &a_x+\mathrm{h.c.} \Big)^2 \\ &\lesssim \frac{1}{N} (\cN+1)^3 \, (1- {\bf 1}_{\vartheta N} (\cN-2))\, .
\end{split}
\end{equation}
As for the second term on the r.h.s. of (\ref{eq:diffe}), its square can be estimated as follows.
\begin{equation}
\begin{split}
\Big( \frac{1}{N}\int\rd x \,\rd y\,&\frac{1}{|x-y|+\alpha_N}  \,a_x^* a_y^* \,(1-{\bf 1}_{\vartheta N} (\cN))\,a_y a_x \Big)^2 \\ = \; &\frac{1}{N^2}\int\rd x\,\rd y\,\rd x'\rd y'\,  \frac{1}{|x-y|+\alpha_N} \, \frac{1}{|x'-y'|+\alpha_N} \, \\ & \hspace{1cm} \times (1-{\bf 1}_{\vartheta N}(\cN-2))  a_x^* a_y^* a_y a_x a_{x'}^* a_{y'}^* a_{y'} a_{x'}(1-{\bf 1}_{\vartheta N}(\cN-2)) \,.
\end{split}
\end{equation}
From $a_y a_x a_{x'}^* a_{y'}^* = a_{x'}^* a_{y'}^* a_y a_x + [a_y a_x , a_{x'}^* a_{y'}^*]$, and from \eqref{eq:comm1}, we conclude that
\begin{equation}\label{L-L2almostfinal}
\begin{split}
\Big( \frac{1}{N}\int\rd x \,\rd y\,&\frac{1}{|x-y|+\alpha_N}  \,a_x^* a_y^* \,(1-{\bf 1}_{\vartheta N} (\cN))\,a_y a_x \Big)^2 \\ \lesssim \; & \frac{1}{\;N^2}\int\rd x\,\rd y\,\rd x'\rd y' \, \frac{1}{|x-y|+\alpha_N} \frac{1}{|x'-y'| +\alpha_N} \\
& \qquad \qquad \times (1-{\bf 1}_{\vartheta N} (\cN-2))  a_x^* a_y^* a_{x'}^* a_{y'}^*  a_{y'} a_{x'} a_y a_x (1-{\bf 1}_{\vartheta N} (\cN-2))  \\
&+ \frac{1}{\;N^2} \int\rd x\,\rd y\,\rd y' \, \frac{1}{|x-y|+\alpha_N} \frac{1}{|x-y'| +\alpha_N} \\
& \qquad \qquad \times (1-{\bf 1}_{\vartheta N} (\cN-2))\, a_x^* a_y^* a_{y'}^*  a_{y'} a_y a_x \, (1-{\bf 1}_{\vartheta N} (\cN-2))  \\
&+ \frac{1}{\;N^2}\int\rd x\,\rd y\, \frac{1}{(|x-y|+\alpha_N)^2} \, (1-{\bf 1}_{\vartheta N} (\cN-2)) \,  a_x^* a_y^* a_y a_x \, (1-{\bf 1}_{\vartheta N} (\cN-2))\\ 
\lesssim \; &\frac{1}{N^2\alpha_N^2} \, (\cN+1)^4  \, (1-{\bf 1}_{\vartheta N} (\cN-2))
\end{split}
\end{equation}
{F}rom (\ref{eq:cubic}) and (\ref{L-L2almostfinal}), we find that
\begin{equation}\label{eq:llsq} (\tLL-\LL)^2 \lesssim \left(\frac{1}{N} + \frac{1}{N^2\alpha_N^2} \right)  (\cN+1)^4 (1-{\bf 1}_{\vartheta N} (\cN-2)) \end{equation}
For every $k \in \bN$, we have $(1-{\bf 1}_{\vartheta N} (\cN-2)) \leq (\cN-2)^k / (\vartheta N)^k$. Therefore
\begin{equation}\label{eq:diff-fin} (\tLL-\LL)^2 \lesssim \left(\frac{1}{N} + \frac{1}{N^2\alpha_N^2} \right)  \frac{(\cN+1)^{k+4}}{(\vartheta N)^k} \, \,. \end{equation}
Analogously to Proposition \ref{prop:ntildebd} (see Lemma 3.8 and its proof in \cite{RS}), there is $C=C(k,\kappa,T, \|\ph \|_{H^2})$ such that 
\[ \langle \wt{\cW}_N (t;0) \Omega, (\cN+1)^{4+k} \, \wt{\cW}_N (t;0) \Omega \rangle \leq C \] for all $|t| \leq T$. {F}rom (\ref{eq:llsq}), we find that, for every $k \in \bN$, there exists $C=C(\vartheta, k, T, \kappa, \| \ph \|_{H^2})$ such that 
\[ \left\|  (\tLL-\LL) \wt{\cW}_N (t;0) \Omega \right\| \leq \frac{C}{N^k} \left(\frac{1}{N} + \frac{1}{N\alpha_N} \right) \, . \]
The proposition now follows from (\ref{cook_U-Utilde}).
 \qed


\thebibliography{hhh}

\bibitem{BGM} Bardos, C.; Golse, F.; Mauser, N. Weak coupling limit of the
$N$-particle Schr\"odinger equation. \textit{Methods Appl. Anal.}
\textbf{7} (2000), 275-293.

\bibitem{ES} Elgart, A.; Schlein, B.: Mean field dynamics of boson stars. {\it Comm. Pure Appl. Math.} {\bf 60} (2007), no. 4, 500-545.

\bibitem{ErS} Erd\H os, L.; Schlein, B.: Quantum dynamics with mean field interactions: a new approach. {\em J. Stat. Phys.} {\bf 134} (2009), no. 5, 859-870.

\bibitem{ESY1} Erd{\H{o}}s, L.; Schlein, B.; Yau, H.-T.:
Derivation of the cubic nonlinear Schr\"odinger equation from
quantum dynamics of many-body systems. {\it Invent. Math.} {\bf 167} (2007), 515-614.

\bibitem{ESY2} Erd{\H{o}}s, L.; Schlein, B.; Yau, H.-T.: Derivation of the Gross-Pitaevskii equation for the dynamics of Bose-Einstein condensate. Preprint arXiv:math-ph/0606017. To appear in {\it Ann. Math.}

\bibitem{ESY3} Erd\H os, L.; Schlein, B.; Yau, H.-T.: Rigorous derivation of the Gross-Pitaevskii equation. {\it Phys. Rev Lett.} {\bf 98} (2007), no. 4, 040404.

\bibitem{ESY4}  Erd{\H{o}}s, L.; Schlein, B.; Yau, H.-T.: Rigorous derivation of the Gross-Pitaevskii equation with a large interaction potential. Preprint arXiv:0802.3877. To appear in {\it J. Amer. Math. Soc.}

\bibitem{EY} Erd{\H{o}}s, L.; Yau, H.-T.: Derivation
of the nonlinear {S}chr\"odinger equation from a many body {C}oulomb
system. \textit{Adv. Theor. Math. Phys.} \textbf{5} (2001), no. 6,
1169--1205.

\bibitem{FL}
Fr\"ohlich, J.; Lenzmann, E.: Blowup for nonlinear wave equations describing bosons stars. {\it Comm. Pure Appl. Math.} {\bf 60} (2007), no. 11, 1691--1705.

\bibitem{FKP}
Fr\"ohlich, J.;Knowles, A.;Pizzo,A.: Atomism and quantization. {\it J. Phys. A: Math. Theor.} {\bf 40} (2007), 3033--3045.

\bibitem{GMM}
Grillakis, M.; Machedon, M.; Margetis, D.: Second-order corrections to mean field evolution of weakly interacting bosons. I. {\it Comm. Math. Phys.} {\bf 294} (2010), no. 1, 273--301.

\bibitem{GMM2}
Grillakis, M.; Machedon, M.; Margetis, D.: Second-order corrections to mean field evolution of weakly interacting bosons. I. Preprint arXiv:1003.4713.

\bibitem{GV} Ginibre, J.; Velo, G.: The classical
field limit of scattering theory for non-relativistic many-boson
systems. I and II. \textit{Commun. Math. Phys.} \textbf{66} (1979),
37--76, and \textbf{68} (1979), 45--68.

\bibitem{GK} Gulisashvili, A.; Kon, M. K.: Exact smoothing properties of Schr\"odinger semigroups, {\it Amer. J. Math.} {\bf 118} (1996), 1215-1248.

\bibitem{H} Hepp, K.: The classical limit for quantum mechanical
correlation functions. \textit{Commun. Math. Phys.} \textbf{35}
(1974), 265--277.

\bibitem{L}
Lenzmann, E.: Well-posedness for semi-relativistic Hartree equations of critical type. 
{\it Math. Phys. Anal. Geom.} {\bf 10} (2007), no. 1, 43--64.

\bibitem{LY}
Lieb, E. H.; Yau, H.-T.: The {C}handrasekhar theory of stellar
collapse as the limit of quantum mechanics. \textit{Comm. Math.
Phys.} \textbf{112} (1987), no. 1, 147--174.

\bibitem{LT}
Lieb, E. H.; Thirring, W. E.: Gravitational collapse in quantum mechanics with relativistic kinetic energy. {\it Ann. Phys.} {\bf 155} (1984), 494--512. 

\bibitem{KP} Knowles, A.; Pickl, P.: Mean-field dynamics: singular potentials and rate of convergence. Preprint arXiv:0907.4313.

\bibitem{P} Pickl, P.: Derivation of the time dependent Gross Pitaevskii equation with external fields. Preprint arXiv:1001.4894.

\bibitem{RS}
Rodnianski, I.; Schlein, B.: Quantum fluctuations and rate of convergence towards mean field dynamics. {\it Comm. Math. Phys.} {\bf 291} (2009), no. 1, 31--61.

\bibitem{S} Spohn, H.: Kinetic equations from Hamiltonian dynamics.
   \textit{Rev. Mod. Phys.} \textbf{52} (1980), no. 3, 569--615.
\end{document}